\documentclass[11pt, letterpaper]{article}
\pdfoutput=1
\usepackage[authoryear]{natbib}
\usepackage{amsmath,amssymb,amsfonts,amsthm,bbm}
\usepackage{epic,eepic,epsfig,longtable}
\usepackage{multirow,verbatim}
\usepackage{array}

\usepackage[T1]{fontenc}

\RequirePackage[colorlinks,citecolor=blue,urlcolor=red]{hyperref}
\usepackage{xr}

\usepackage{epsfig}
\usepackage{color}

\makeatletter
\def\singlespace{\def\baselinestretch{1}\@normalsize}

\makeatletter
\def\singlespace{\def\baselinestretch{1}\@normalsize}

\numberwithin{equation}{section}

\renewcommand{\hat}{\widehat}

\renewcommand{\hat}{\widehat}

\newcommand{\bfm}[1]{\ensuremath{\mathbf{#1}}}

   \def\bA{\bfm A}  
\def\bb{\bfm b}   \def\bB{\bfm B}  
   \def\bC{\bfm C}

\def\bff{\bfm f}  \def\bF{\bfm F}  
     
\def\bh{\bfm h}   \def\bH{\bfm H}  
   \def\bI{\bfm I}  
     
   \def\bK{\bfm K}  
     
   \def\bM{\bfm M}

   \def\bP{\bfm P}

\def\bu{\bfm u}   \def\bU{\bfm U}  
\def\bv{\bfm v}     
\def\bw{\bfm w}   \def\bW{\bfm W}  
\def\bx{\bfm x}   \def\bX{\bfm X}  
\def\by{\bfm y}   \def\by{\bfm y}  
\def\bz{\bfm z}   \def\bZ{\bfm Z}

\newcommand{\bfsym}[1]{\ensuremath{\boldsymbol{#1}}}

 \def\balpha{\bfsym \alpha}
 \def\bbeta{\bfsym \beta}
 \def\bgamma{\bfsym \gamma}             
            
 \def\bfeta{\bfsym {\eta}}              
 \def\bmu{\bfsym {\mu}}                 
 	
 \def\btheta{\bfsym {\theta}}           
 \def\beps{\bfsym \varepsilon}          \def\bepsilon{\bfsym \varepsilon}
              \def\bSigma{\bfsym \Sigma}
              \def\bPhi{\bfsym \Phi}
         \def\bLambda {\bfsym {\Lambda}}

 \def\bxi{\bfsym {\xi}}
 
\DeclareMathOperator*{\argmin}{argmin}
\DeclareMathOperator*{\argmax}{argmax}

\DeclareMathOperator{\cov}{cov}

\DeclareMathOperator{\diag}{diag}

\DeclareMathOperator{\E}{E}

\DeclareMathOperator{\supp}{supp}

\DeclareMathOperator{\sign}{\rm sign}

\def\eps{\varepsilon}
\def\beps{\mbox{\boldmath$\eps$}}
\def\newpage{\vfill\eject}

\def\today{\ifcase\month\or
  January\or February\or March\or April\or May\or June\or
  July\or August\or September\or October\or November\or December\fi
  \space\number\day, \number\year}

\newdimen\biblioindent    \biblioindent=30pt

 at 8truept

\def\eps{\varepsilon}

\newcommand{\beq}{\begin{equation}}
  \newcommand{\eeq}{\end{equation}}
\newcommand{\beqn}{\begin{eqnarray}}
  \newcommand{\eeqn}{\end{eqnarray}}
\newcommand{\beqnn}{\begin{eqnarray*}}
  \newcommand{\eeqnn}{\end{eqnarray*}}

\newcommand{\etal}{{\it et al.\ }}

\allowdisplaybreaks
\usepackage{verbatim}
\renewcommand{\baselinestretch}{1.4}
\numberwithin{equation}{section}
\theoremstyle{plain}
\newtheorem{thm}{Theorem}[section]
\newtheorem{defn}{Definition}[section]
\newtheorem{lem}{Lemma}[section]
\newtheorem{cor}{Corollary}[section]

\newtheorem{ass}{Assumption}[section]
\theoremstyle{definition}

\newtheorem{rem}{Remark}[section]

\newcounter{CondCounter}

\def \bbP      {\mathbb{P}}

\def \R       {\mathbb{R}}
\def \M {\mathcal{M}}

\begin{document}
\renewcommand{\baselinestretch}{1.5}

\title{\bf Factor-Adjusted Regularized Model Selection }
    \author{Jianqing Fan
  \hspace{.2cm}\\
    Department of ORFE, Princeton University\\
    Yuan Ke \thanks{Corresponding author. Address: 310 Herty Drive
University of Georgia, Athens, GA 30602, USA. Phone: 609-955-8395. Fax: 706-542-3391. Email: yuan.ke@uga.edu.}\\
    Department of Statistics, University of Georgia\\
        and \\
        Kaizheng Wang\\
        Department of ORFE, Princeton University
        }
\date{}

\maketitle
\begin{abstract}
This paper studies model selection consistency for high dimensional sparse regression when data exhibits both cross-sectional and serial dependency.
Most commonly-used model selection methods fail to consistently recover the true model when the covariates are highly correlated.
Motivated by econometric studies, we consider the case where covariate dependence can be reduced through factor model, and propose a consistent strategy named Factor-Adjusted Regularized Model Selection (FarmSelect).
By separating the latent factors from idiosyncratic components, we transform the problem from model selection with highly correlated covariates to that with weakly correlated variables. Model selection consistency as well as optimal rates of convergence are obtained under mild conditions. Numerical studies demonstrate the nice finite sample performance in terms of both model selection and out-of-sample prediction. Moreover, our method is flexible in a sense that it pays no price for weakly correlated and uncorrelated cases.
Our method is applicable to a wide range of high dimensional sparse regression problems. An R-package {\em FarmSelect} is also provided for implementation.
\end{abstract}


\renewcommand{\baselinestretch}{1.5}
\section{Introduction}


Specifying an appropriate yet parsimonious model is a key topic in economics and statistics studies. Parsimonious models are preferable due to their simplicity and interpretability. In addition, removing redundent coefficients can improve the prediction accuracy. In classic econometric studies, extensive efforts have been made to identify the correct orders of time series models, see \cite{AIC}, \cite{BIC}, \cite{TT85}, \cite{Choi92} and \cite{TT89} among others. With the development of data collection and  storage technologies, high dimensional time series  characterize  many  contemporary  research problems in economics, finance, statistics, machine learning and so on.
Therefore, over the past two decades, many model selection methods have been developed.
A major part of them are based on the regularized $M$-estimation approach including 
the LASSO \citep{Tibshirani_1996}, the SCAD \citep{Fan_Li_2001}, the elastic net \citep{Zou_Hastie_2005}, and the Dantzig selector \citep{candes_tao_2007}, among others.
These methods have attracted a large amount of theoretical and algorithmic
studies. See \cite{Donoho_Elad_2003}, \cite{Fan_Peng_2004}, \cite{Efron_04},
\cite{M_B_2006}, \cite{Zhao_Yu_2006}, \cite{Fan_Lv_2008}, \cite{Zou_Li_2008}, \cite{Bickel_09}, \cite{Wainwright2009}, \cite{Zhang_2010}, and references therein.
However, most existing model selection schemes are not tailored for economic and finance applications as they assume covariates are cross-sectionally weakly correlated and serially independent. These conditions are easily violated in economic and financial datasets. For example, economics studies \citep[e.g.][]{SW02, Bai_Ng_02} show that there exist strong co-movements among a large pool of macroeconomic variables. A stylized feature of the stock return data is cross-sectionally correlated among the stock returns. Furthermore, even if the weakly correlated assumption holds, one may still observe strong spurious correlations in a high dimensional sample.

\begin{figure}[htbp]
 \centering
 \includegraphics[width=6.5 in]{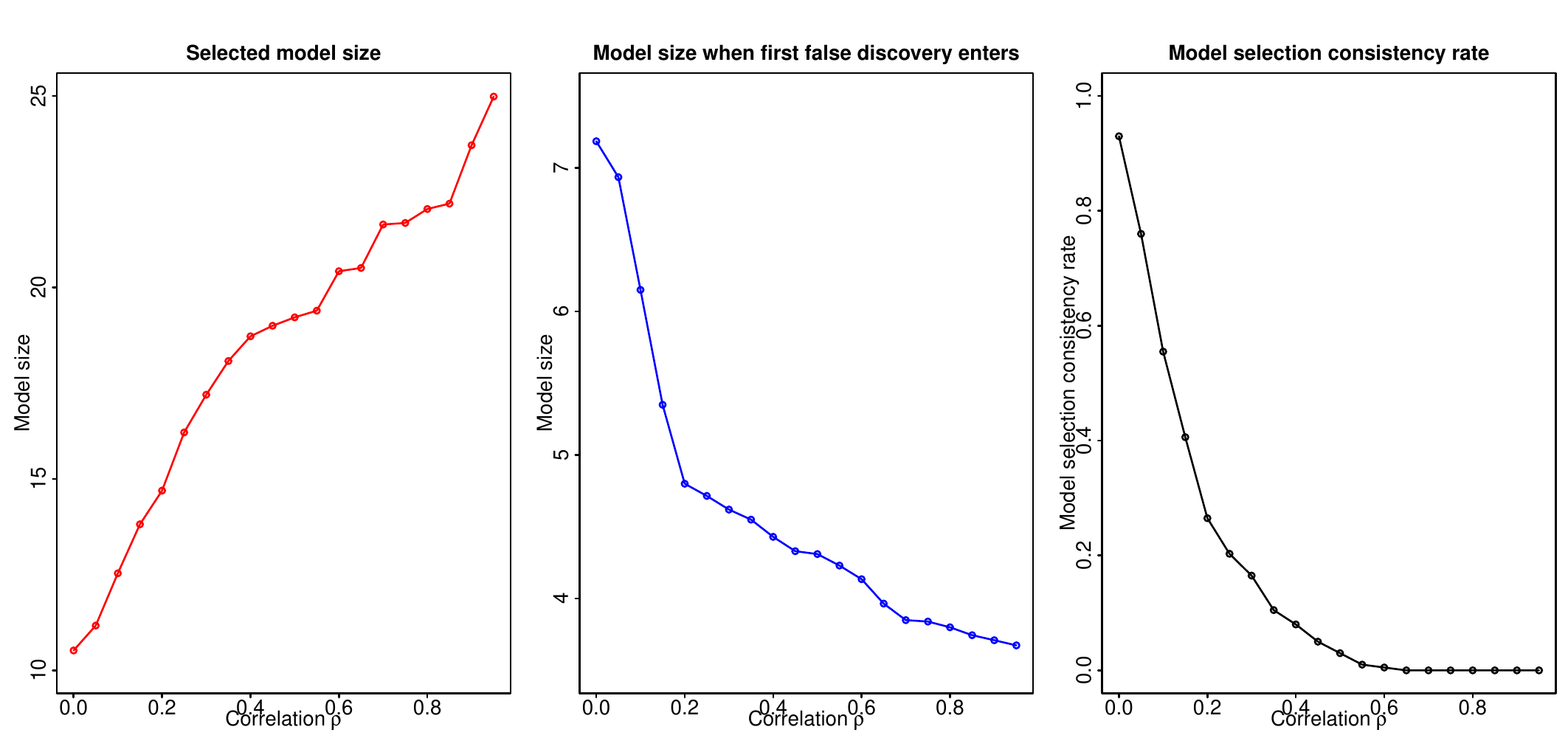}
  \caption{LASSO model selection results with respect to the correlations}
 \label{Fig_corr_1}
\end{figure}

To illustrate how cross-sectional correlations influence the model selection result, we consider a toy example of LASSO with an equally correlated design.
Consider a sparse linear model $\by=\bX \bbeta^{\ast}+\bepsilon$ with no intercept. We choose sample size $n=100$, dimensionality $p=200$, $\bbeta^{\ast}=(\beta_1, \cdots , \beta_{10}, {\bf 0}^{T}_{(p-10)})^{T}$, and $\bepsilon \sim N({\bf 0}_n, \ 0.3 {\bf I}_n)$.  The nonzero coefficients $\beta_1, \cdots , \beta_{10}$ are drawn from i.i.d. Uniform $[2,5]$.
The covariates $\bX=(\bx_1,\cdots,\bx_p)^{T}$ are drawn from the normal distribution $N({\bf 0}_p, \ \bSigma)$ where $\bSigma$ is a correlation matrix with all off-diagonal elements $\rho$ for some $\rho \in [0,\ 1)$.
Let $\rho$ increase from 0 to $0.95$ by a step size $0.05$. For each given $\rho$, we simulate 200 replications and calculate the average model size selected by LASSO, the average model size when the first false discovery ($\bx_j$, $j>10$) enters the solution path and the  model selection consistency rate. As is shown in Figure \ref{Fig_corr_1}, the correlation influences the model selection results in the following three aspects:
(i) selected model size, (ii) early selection of false variables, (iii) model selection consistency rates.
Therefore, when the covariates are highly correlated, there is little hope to exactly recover the active set from the solution path of LASSO. As to be shown later, the correlation has similar adverse impacts on other model selection methods (e.g. SCAD and elastic net).

To overcome the the aforementioned problems caused by the cross-sectional correlation, this paper proposes a consistent strategy named Factor-Adjusted Regularized Model Selection (FarmSelect) for the case where covariates can be decorrelated via a few pervasive latent factors. More precisely, let $x_{tj}$ be the $t$th ($t=1, \cdots , n$) observation of the $j$th ($j=1, \cdots , p$) covariate,
and assume that $\bx_t=(x_{t1}, \cdots , x_{tp})^T$ follows an approximate factor model
\begin{equation}\label{eq_1.1}
\bx_t=\bB \bff_t + \bu_t,
\end{equation}
where $\bff_t$ is a $K \times 1$ vector of latent factors, $\bB$ is a $p \times K$ matrix of factor loadings, and $\bu_t$ is a $p \times 1$ vector of idiosyncratic components that are uncorrelated with $\bff_t$. The strategy of FarmSelect is to first learn the parameters in approximate factor model (\ref{eq_1.1}) for the covariates $\{\bx_t\}_{t=1}^n$. Denote by $\hat{\bff}_t$ and $\hat{\bB}$ the obtained estimators of the factors and loadings respectively. Then by identifying the highly correlated low rank part by $\hat{\bB}\hat{\bff}_t$, we transform the problem from model selection with highly correlated covariates in $\bx_t$ to that with weakly correlated or uncorrelated  idiosyncratic components $\hat{\bu}_t:=\bx_t-\hat{\bB}\hat{\bff}_t$ and $\hat{\bff}_t$. The second step amounts to solving a regularized profile likelihood problem. We study FarmSelect in details by providing theoretical guarantees that FarmSelect can achieve model selection consistency as well as estimation consistency under mild conditions. Unlike traditional studies in model selection where the samples are assumed to be i.i.d., serial dependency is allowed and thus our theories apply to time series data. Moreover, both theoretical and numerical studies show the flexibility of FarmSelect in a sense that it pays no price for weakly correlated cases. This property makes FarmSelect very powerful when the underlying correlations between active and inactive covariates are unknown.

FarmSelect is applicable to a wide range of high dimensional sparse regression related problems that include but are not limited to linear model, generalized linear model, Gaussian graphic model, robust linear model and group LASSO.
For the sparse linear regression, the proposed approach is equivalent to projecting the response variable and covariates onto
the linear space orthogonal to the one spanned by the estimated factors. Existing algorithms that yield solution paths of LASSO can be directly applied in the second step.
To demonstrate the finite sample performance of FarmSelect, we study two simulated and one empirical examples.
The numerical results show FarmSelect can consistently select the true model even when the covariates are highly correlated
while existing methods like LASSO, SCAD and elastic net fail to do so. An R-package FarmSelect (
\url{https://cran.r-project.org/web/packages/FarmSelect}
) is also provided to facilitate the  implemention our method.

Various methods have been studied to estimate the approximate factor model.
Principal components analysis \citep[PCA, ][]{SW02} is among one of the most popular ones.
Data-driven estimation methods of the number of factors have been studied in extensive literature, such as \cite{Bai_Ng_02}, \cite{Luo09}, \cite{HL},
\cite{LamYao}, and \cite{Ahn_Horen_2013} among others. Recently, a large amount of literature contributed to the asymptotic analysis of PCA under the ultra-high dimensional regime including
\cite{Johnstone_Lu_2009}, \cite{POET}, \cite{Shen_16} and \cite{WFa17}, among others.

The rest of the paper is organized as follows. Section 2 overviews the problem setup including regularized $ M $-estimator of sparse regression, the {\it irrepresentable condition} and approximate factor models. Section 3 introduces the model selection methodology of FarmSelect and studies  the sparse generalized linear model as a showcase example.
Some issues related to the estimation of approximate factor models will be discussed in Section 3 as well.   Section 4 presents the general theoretical results.  Section 5 provides simulation studies and Section 6 studies the forecast of U.S. bond risk premia.
The appendix contains the technical proofs.

Here are some notations that will be used throughout the paper. $\bI_n$ denotes the $n\times n$ identity matrix; $\mathbf{0}$ refers to the $n\times m$ zero matrix; $\mathbf{0}_n$ and $\mathbf{1}_n$ represent the all-zero and all-one vectors in $\R^n$, respectively.
 For a matrix $\bM$, we denote
its matrix entry-wise max norm as $\|\bM\|_{\max}=\max_{i,j}|M_{ij}|$ and denote by $\|\bM\|_F$ and $\|\bM\|_p$  its Frobenius and induced $p$-norms, respectively. $\lambda_{\min}(\bM)$ denotes the minimum eigenvalue of $\bM$ if it is symmetric. For $\bM\in\R^{n\times m}$, $I\subseteq[n]$ and $J\subseteq [m]$, define $\bM_{IJ}=(\bM_{ij})_{i\in I,j\in J}$, $\bM_{I\cdot}=(\bM_{ij})_{i\in I,j\in [m]}$ and $\bM_{\cdot J}=(\bM_{ij})_{i\in [n],j\in J}$. For a vector $\bv\in\R^p$ and $S\subseteq[p]$, define $\bv_S=(\bv_{i})_{i\in S}$ to be its subvector. Let $\nabla$ and $\nabla^2$ be the gradient and Hessian operators. For $f:\R^p\rightarrow\R$ and $I,J\in[p]$, define $\nabla_I f(x)=(\nabla f(x))_I$ and $\nabla^2_{IJ}f(x)=(\nabla^2 f(x))_{IJ}$. $N(\bmu,\bSigma)$ refers to the normal distribution with mean $\bmu$ and covariance matrix $\bSigma$.

\section{ Problem Setup}
\subsection{Regularized $ M $-estimator}\label{sec_M_estimator}
Let us begin with a family of high dimensional sparse regression problems in the following settings. From now on we suppose that $\{ \bx_t \}_{t=1}^n$ are $(p-1)$-dimensional random vectors of covariates with zero mean\footnote{We use $(p-1)$ instead of $p$ to denote the number of covariates so that there are $p$ coefficients including the intercept. In addition, we center the covariates if they could have non-zero means.  Whether this step is done or not will not affect the estimation of $\{\bbeta_j^*\}_{j=1}^p$, but affect the intercept $\beta_0^*$.}, and $\{ y_t \}_{t=1}^n$ are responses with each $y_t$ sampled from some probability distribution ${\mathbb P}({z_t})$ parametrized by $z_t = \beta_0^*+\sum_{j=1}^{p-1}\beta^{\ast}_j \bx_{tj} = (1,\bx_t^T) \bbeta^*$. Here $\bbeta^*=(\beta_0^*,\cdots,\beta_{p-1}^*)^T \in \R^{p}$ is a sparse vector with $s\ll p$ non-zero elements.
Let $\bX=(\bx_1, \ \cdots, \bx_n)^T \in \R^{n \times (p-1)}$ and $\by=(y_1, \ \cdots y_n)^{T} \in \R^n$ be the design matrix and response vector, respectively.
Define $\bX_1=(\mathbf{1}_n,\bX) \in \R^{n\times p}$, where the subscript $1$ refers to the all-one column added to the original design matrix $\bX$.

Let $ L_n (\by, \bX_1\bbeta)$ be some convex and differentiable loss function assigning a cost to any parameter $\bbeta \in \R^{p}$. Suppose that $\bbeta^*$ is the unique minimizer of the population risk $\E [ L_n (\by, \bX_1 \bbeta) ]$. Under the high-dimensional regime, it is natural to estimate $\bbeta^{\ast}$ via a regularized $ M $-estimator as follows:
\begin{equation}\label{M_estimator}
\widetilde{\bbeta} \in \argmin\limits_{ \bbeta \in \R^{p} }
\left\{
L_n (\by,\bX_{1} \bbeta)+\lambda R_n(\bbeta)\right\},
\end{equation}
where $R_n : \R^{p} \rightarrow \R_{+}$ is a norm that penalizes the use of a nonsparse vector $\bbeta$ and $\lambda>0$ is a tuning parameter.

A special case of this problem is the $L^1$ penalized likelihood estimation of generalized linear models. Suppose the conditional density function of $Y$ given covariates $\bx$ is a member of the exponential family, i.e.
\begin{equation}\label{glm}
f(y|\bx,\bbeta^{\ast})\propto \exp [ y z-b(z)+c(y) ],
\end{equation}
where $z=\beta^*_0 + \sum_{j=1}^{p-1} \beta_j^* x_j =(1,\bx^T) \bbeta^*$, $b(\cdot)$ and $c(\cdot)$ are known functions, and $\bbeta^{\ast}$ is an unknown coefficient vector of interest. It is commonly assumed that $b(\cdot)$ is strictly convex.
Taking the loss function to be the negative log-likelihood function and the penality function to be the $L^1$ norm, the regularized $ M $-estimator of $\bbeta^{\ast}$ admits the form
\begin{equation}\label{traditional-L1-PMLE}
\widetilde{\bbeta} \in
\argmin_{\bbeta\in\R^{p}}
\left\{\frac{1}{n}\sum_{t=1}^{n}[-y_t (1,\bx_t^{T})\bbeta+b( (1,\bx_t^{T})\bbeta)]+\lambda \|\bbeta\|_1 \right\}.
\end{equation}

\subsection{Irrepresentable condition}

We expect a good estimator of (\ref{M_estimator}) to achieve estimation consistency as well as selection consistency. The former one requires $\|\hat{\bbeta}-\bbeta^{\ast} \| \xrightarrow{\mathrm{P}} {\bf 0}$ for some norm $\|\cdot\|$ as $n \rightarrow \infty$; while the latter one requires $\mathrm{P} ( \text{supp}(\hat{\bbeta}) = \text{supp}(\bbeta^{\ast}) ) \rightarrow 1$ as $n \rightarrow \infty$.
In general, the estimation consistency does not imply the selection consistency and vice versa. To study the selection consistency,
we consider a stronger condition named general sign consistency as follows.

{\defn[Sign consistency] An estimate $\hat{\bbeta}$ is  sign consistent with respect to $\bbeta^{\ast}$ if $\exists \lambda \geq 0$ such that
$ \lim\limits_{n \rightarrow \infty} \mathrm{P} (\mathrm{sign}(\hat{\bbeta})= \mathrm{sign}(\bbeta^{\ast}))=1$.
}

\cite{Zhao_Yu_2006} studied the  LASSO estimator and showed there exists an {\it irrepresentable condition}  which is sufficient and almost necessary for both  sign and estimation consistencies for sparse linear model.
Without loss of generality, we assume $\supp(\bbeta^*)=[s]=S$. Denote $(\bX_1)_{S}$ and $(\bX_1)_{S^c}$ as the submatrices of $\bX_1$ defined by its first $s$ columns and the rest $(p-s)$ columns, respectively.
Then the {\it irrepresentable condition} requires some $\tau\in(0,1)$, such that
\begin{equation}\label{IC}
\| (\bX_1)_{S^c}^{T} (\bX_1)_{S} [ (\bX_1)_{S}^T (\bX_1)_{S} ]^{-1}\|_{\infty}\leq 1-\tau.
\end{equation}

For general regularized $ M $-estimator (\ref{M_estimator}) to achieve both sign and estimation consistencies, \cite{Lee15} proposed a generalized {\it  irrepresentable condition}. When applied to the $L^1$ regularizer, it becomes
\begin{equation}\label{general_IC}
\|\nabla^2_{S^c S}L(\bbeta^{\ast}) [ \nabla^2_{S S}L(\bbeta^{\ast}) ]^{-1}\|_{\infty}\leq 1-\tau,
\end{equation}
for some $\tau\in(0,1)$, where $L(\bbeta)=L_n(\by,\bX_1\bbeta)$. It is easy to check (\ref{general_IC}) is equivalent to (\ref{IC}) under the LASSO case.
The generalized {\it  irrepresentable condition} will easily get violated when there exists strong correlations between active and inactive variables. Even if it holds, the key parameter $\tau$ can be very close to zero, making it hard to select the correct model and obtain small estimation errors simultaneously.

\subsection{Approximate factor model}

To go beyond the assumption on weakly correlation, a natural extension is conditional weak correlation. Suppose covariates are dependent through latent common factors.
Given these common factors, the idiosyncratic components are weakly correlated. Factor model has been well studied in econometrics and statistics literature, we refer to \cite{LM71, SW02, Bai_Ng_02, Forni_2005, POET}, among others.

We assume that $\{ \bx_t \}_{t=1}^n \subseteq \R^{p-1}$ follows the approximate factor model
\begin{equation}\label{eq_factor}
\bx_t = \bB \bff_t + \bu_t,\qquad t\in[n],
\end{equation}
where $\{ \bff_t \}_{t=1}^n \subseteq \R^K$ are latent factors, $\bB\in\R^{(p-1)\times K}$ is a loading matrix, and $\{ \bu_t \}_{t=1}^n \subseteq \R^{p-1}$ are idiosyncratic components. Note that $\bx_t$ is the only observable quantity. Throughout the paper, $K$ is assumed to be independent of $n$, which is a standard assumption in the literature of factor model \cite{POET}.
We assume that $\{ \bff_t,\bu_t \}_{t=1}^n$ comes from a time series $\{ \bff_t,\bu_t \}_{t=-\infty}^{\infty}$. Denote $\bF=(\bff_1,\ \cdots,\ \bff_n)^{T}\in\R^{n \times K}$ and $\bU=(\bu_1,\cdots,\bu_n)^{T} \in \R^{n \times (p-1)}$. Then (\ref{eq_factor}) can be written in a more compact form:
\begin{equation}\label{factor-x}
\bX=\bF  \bB^{T} +\bU.
\end{equation}
We impose the following identifiability assumption \citep{POET}. Here we only put the most basic assumption for factor model, and more can be found in Section \ref{sec_factor_estimate} where estimation of factor model is discussed.
\begin{ass}\label{assump-factor-1}
Assume that $\cov(\bff_t)=\bI_K$, $\bB^T \bB$ is diagonal, and all the eigenvalues of $\bB^T \bB /p$ are bounded away from 0 and $\infty$ as $p\to\infty$.
\end{ass}

\section{Factor-adjusted regularized model selection }

\subsection{Methodology}\label{sec_method}
To illustrate the main idea, we temporarily assume $\bff_t$ and $\bu_t$ to be observable. Define $\bB_{0}=(\mathbf{0}_K,\bB^T)^T\in\R^{K\times p}$ and $\bU_{1}=(\mathbf{1}_n,\bU)\in \R^{n\times p}$.
By the approximate factor model (\ref{factor-x}), we have decompositions $\bX_1 = \bF \bB_0^T + \bU_1$ and
$$
\bX_1 \bbeta=\bF \bB_0^{T} \bbeta + \bU_1 \bbeta = \bF \bgamma+\bU_1 \bbeta,
$$
where $\bgamma=\bB_0^{T} \bbeta \in \R^K$. The regularized $ M $-estimator (\ref{M_estimator}) can be rewritten as
\begin{equation*}
\widetilde{\bbeta} \in \argmin\limits_{\bbeta \in \R^{p}, \ \bgamma = \bB_0^T \bbeta \in \R^K, }
\left\{
L_n (\by, \bF \bgamma+ \bU_1 \bbeta)+\lambda R_n (\bbeta)\right\}.
\end{equation*}
Instead of using $\widetilde{\bbeta}$ to estimate $\bbeta^*$, we regard $\bgamma$ as nuisance parameters, drop the constraint $\bgamma = \bB_0^T \bbeta$, and consider a new estimator
\begin{equation}\label{Profile_likelihood}
\hat{\bbeta} \in \argmin\limits_{ \bbeta \in \R^{p}, \ \bgamma\in \R^K}
\left\{
L_n (\by, \bF \bgamma+ \bU_1 \bbeta)+\lambda R_n (\bbeta)\right\},
\end{equation}
namely $(\bu_t^T,\bff_t^T)^T$ are now regarded as new covariates. In other words, by lifting the covariate space from $\R^{p}$ to $\R^{p+K}$, the highly dependent covariates $\bx_t$ are replaced by weakly dependent ones.

The theoretical for us to ignore the constraint $\bgamma = \bB_0^T \bbeta$ is given by the following lemma, whose proof is given by Appendix~\ref{general-inverse-problems}.

\begin{lem}\label{lem-identifiability}
Consider the generalized linear model (\ref{glm}), let $L_n ( \by,\bz)=\frac{1}{n}\sum_{t=1}^{n}[-y_t z_t +b( z_t)]$, $\eta_t=y_t - b'( (1,\bx_t^T) \bbeta^*)$ and $\bw_t=(1,\bu_t^T,\bff_t^T)^T$. If $\E (\eta_t \bw_t )=\mathbf{0}_{p+K}$, then
\begin{equation*}
( \bbeta^* , \bB_0^T \bbeta^* ) = \argmin_{  \bbeta \in \R^{p},\ \bgamma \in \R^K } \E [ L_n( \by,\bF \bgamma + \bU_1\bbeta ) ].
\end{equation*}
\end{lem}

It is worth pointing out that the assumption $\E ( \eta_t \bw_t)=\mathbf{0}_{p+K}$ is very mild and natural. We just assume the residual $\eta_t$ and augmented covariates $\bw_t$ to be uncorrelated, which is almost as weak as the standard condition $\E (\eta_t | \bx_t)=0$ for the generalized linear model. For example, in the linear model where $b(t)=t^2/2$ and $y_t = (1,\bx_t^T) \bbeta^* + \eta_t$, we just strengthen the condition from $E(\eta_t \bx_t)=0$ to $E(\eta_t \bff_t)=0$ and $E(\eta_t \bu_t)=0$. In particular, the assumptions hold if $\eta_t$ is independent of $\bu_t$ and $\bff_t$.

By construction, $(\bU,\bF)$ has now much weaker cross-sectional correlation than $\bX$. Thus, the penalized profile likelihood (\ref{Profile_likelihood}) removes the effect of strong correlations caused by the latent factors. It can be implemented as follows:

\quad {\it Step 1: Initial estimation}. Let $\bX \in \R^{n\times p}$ be the design matrix. Fit the approximate factor model (\ref{factor-x}) and denote $\hat{\bB}$, $\hat{\bF}$ and $\hat{\bU}=\bX-\hat{\bF}\hat{\bB}^{T}$ the obtained estimators of $\bB$, $\bF$ and $\bU$ respectively by using the principal component analysis \citep{Bai03, POET}.  More specifically, the columns of $\hat{\bF}/\sqrt{n}$ are the eigenvectors of $\bX\bX^T$ corresponding to the top $K$ eigenvalues, $\hat{\bB} = n^{-1} \bX^T \hat{\bF}$.  This is the same as
 $\hat{\bB}=(\sqrt{\lambda_1}\bxi_1, \cdots,  \sqrt{\lambda_K}\bxi_K)$
    and $\hat{\bF} = \bX \bB \diag(\lambda_1^{-1} \cdots, \lambda_K^{-1} )$, where $\{\lambda_j\}_{j=1}^K$ and $\{\bxi_j\}_{j=1}^K$ are top $K$ eigenvalues in descending order and their associated eigenvectors of the sample covariance matrix.

\vspace{0.1in}

\quad {\it Step 2: Augmented $ M $-estimation}. Define $\hat{\bW}=( \mathbf{1}_n,\ \hat{\bU},\ \hat{\bF}) \in \R^{n\times(p+K)}$ and
$\btheta=(\bbeta^{T}, \bgamma^{T})^T\in\R^{p+K}$. Then $\hat{\bbeta}$ is obtained from the first $p$ entries of the solution to the augmented problem
\begin{equation}\label{L1-ppmle}
\hat{\btheta} \in \argmin\limits_{\btheta\in \R^{p+K}}
 \left\{ L_n (\by, \hat{\bW}\btheta)+\lambda R_n (\btheta_{[p]}) \right\}.
\end{equation}

\vspace{0.1in}

We call the above two-step method as the factor-adjust regularized model selection ({\em FarmSelect}). If $\bu_t$ is independent of $\bff_t$ and the variables in the idiosyncratic component $\bu_t$ are weakly correlated, then the columns in $\hat\bW=(\mathbf{1}_n,\hat\bU,\hat\bF)$ are weakly correlated  as long as $\bF$ and $\bU$ are well estimated.
Hence, we successfully transform the problem from model selection with highly correlated covariates $\bX$ in (\ref{M_estimator}) to model selection with weakly correlated or uncorrelated ones by lifting the space to higher dimension. The augmented problem (\ref{L1-ppmle}) is a convex optimization problem which can be minimized via many existing convex optimization algorithms, for example  coordinate descent \citep[e.g.][]{Friedman_etal_2010}
and  ADMM.

\subsection{Example: sparse linear model}\label{sec_linear}

Now we illustrate the FarmSelect procedure using sparse linear regression, where $\by=\bX_1 \bbeta^{\ast}+\beps$.
By defining $\hat{\bB}_0 = (\mathbf{0}_K,\hat{\bB}^T)^T \in \R^{p\times K}$ and $\hat{\bU}_1=(\mathbf{1}_n, \hat{\bU}) \in \R^{n\times p}$, we have $\bX_1=\hat\bF \hat\bB_0 + \hat\bU_1$ and
\begin{equation}\label{MSC_linear}
\by=\bX_1 \bbeta^{\ast}+\beps=\hat{\bF} \hat{\bB}_0^{T} \bbeta^* + \hat{\bU}_1 \bbeta^*+\beps.
\end{equation}
The augmented $ M $-estimator (\ref{L1-ppmle}) for the sparse linear model is of the following form:
\begin{equation*}\label{MSC_PLS}
\hat{\bbeta} \in \argmin\limits_{\bbeta \in \R^p,\ \bgamma\in\R^K}
\left\{
\frac{1}{2n}\| \by-\hat{\bF} \bgamma- \hat{\bU}_1 \bbeta \|_2^2 + \lambda \|\bbeta\|_1 \right\}.
\end{equation*}

Solving the least-squares problem with respect to $\bgamma$, we have the penalized profile least-squares solution
\begin{equation}\label{example-lasso}
\hat{\bbeta} \in \argmin\limits_{\bbeta \in \R^p}
\left\{
\frac{1}{2n}\| (\bI_n-\hat{\bP})(\by-\hat{\bU}_1 \bbeta) \|_2^2 + \lambda \|\bbeta\|_1 \right\},
\end{equation}
where $\hat{\bP}= \hat{\bF}(\hat{\bF}^T\hat{\bF})^{-1}\hat{\bF}^T$ is
the $n \times n$ projection matrix onto the column space of $\hat{\bF}$. %
As the decorrelation step does not depend on the choice of the regularizer $R(\cdot)$, FarmSelect can be applied to a wide range of penalized least squares problems such as SCAD, group LASSO, elastic net, fused LASSO, folded concave penalty such as SCAD, and so on.


\medskip

There is another way to understand this method. By left multiplying the projection matrix $(\bI_n-\hat{\bP})$ to both sides of (\ref{MSC_linear}), we have
\begin{equation}\label{examplel_lasso_project}
(\bI_n-\hat{\bP}) \by = (\bI_n-\hat{\bP}) \hat{\bU}_1 \bbeta^*+(\bI_n-\hat{\bP})\beps,
\end{equation}
where  $(\bI_n-\hat{\bP})\hat{\bU}_1 $ can be treated as the decorrelated design matrix and $(\bI_n-\hat{\bP}) \by$ is the corresponding response variable. From (\ref{examplel_lasso_project}) we see that the method in \cite{Kneip_Sarda_2011} coincides with FarmSelect in the linear case. However, the projection-based representation only makes sense in sparse linear regression. In contrast, our idea of profile likelihood directly generalizes to more general problems.

\medskip

\vspace{-0.05in}

\subsection{Estimating factor models}\label{sec_factor_estimate}
Principal component analysis (PCA) is frequently used to estimate latent factors for model (\ref{factor-x}).
The estimated matrix of latent factors $\hat{\bF}$ is  $\sqrt{n}$ times the eigenvectors corresponding to the $K$ largest eigenvalues of the $n \times n$ matrix $\bX \bX^{T}$. Using the normalization $\bF^{T} {\bF} /n=\bI_{K}$ yields
$\hat{\bB}=\bX^T \hat{\bF} /n$.
Now we introduce the asymptotic properties of  estimated factors and idiosyncratic components. We adopt the regularity assumptions in \cite{POET}, which are similar to the ones in \cite{Bai03} and other literature on high-dimensional factor analysis.

\begin{ass}\label{assump-factor-2}
	\begin{enumerate}
		\item $\{ \bff_t,\bu_t \}_{t=1}^{\infty}$ is strictly stationary and in addition, $\E f_{tk} = \E u_{tj}=\E ( u_{tj}f_{tk})=0$ for all $i\in[n]$, $j\in[p-1]$ and $k\in[K]$;
		\item There exist constants $c_1, c_2 > 0$ such that $\lambda_{\min} ( \cov(\bu_t) ) > c_1$, $\| \cov (\bu_t) \|_1 < c_2$ and \\ $\min_{j,k\in[p-1]} \mathrm{var}( u_{tj} u_{tk} ) >c_1$;
		\item There exist $r_1,r_2>0$ and $b_1,b_2>0$ such that for any $s>0$, $j\in[p-1]$ and $k\in[K]$, $\mathrm{P}( |u_{tj}|>s ) \leq \exp( -(s/b_1)^{r_1} )$ and $\mathrm{P}( |f_{tk}|>s ) \leq \exp( -(s/b_2)^{r_2} )$.
	\end{enumerate}
\end{ass}

\begin{ass}\label{assump-factor-3}
	Let $\mathcal{F}_{-\infty}^{0}$ and $\mathcal{F}_{T}^{\infty}$ denote the $\sigma$-algebras generated by $\{ (\bff_t,\bu_t): i\leq 0\}$ and $\{ (\bff_t,\bu_t): i\geq T\}$ respectively. Assume the existence of $r_3,C>0$ such that $3/r_1+3/(2r_2)+1/r_3>1$ and for all $T\geq 1$,
	\[
	\sup_{
		A\in\mathcal{F}_{-\infty}^{0},B\in \mathcal{F}_{T}^{\infty}
	} |\mathrm{P}(A)\mathrm{P}(B)-\mathrm{P}(AB)|
	\leq \exp(-CT^{r_3});
	\]
\end{ass}

\begin{ass}\label{assump-factor-4}
	There exists $M>0$ such that for all $t,s\in[n]$, we have
	$\| \bB \|_{\max}<M$, $\E \{
	p^{-1/2} [ \bu_t^T \bu_s - \E ( \bu_t^T \bu_s ) ]^4
	\}<M$ and
	$\E \| p^{-1/2}  \bB^T \bu_{t} \|_2^4<M$.
\end{ass}
	
We summarize useful properties of $\hat\bF$ and $\hat\bU$ in Lemma \ref{lem-factor}, which directly follows from Lemmas 10-12 in \cite{POET}.
\begin{lem}\label{lem-factor}
Let $\gamma^{-1} = 3/r_1 + 3/(2r_2) + 1/r_3 +1$.
Suppose that $\log p = o(n^{\gamma/6})$, $n=o(p^2)$, and Assumptions \ref{assump-factor-1}, \ref{assump-factor-2}, \ref{assump-factor-3} and \ref{assump-factor-4} hold. There exists a nonsingular matrix $\bH_0 \in \R^{K\times K}$ such that
\begin{enumerate}
\item $\|\hat{\bF} \bH_0-\bF\|_{\max}=O_{\mathrm{P}}( \frac{1}{\sqrt{n}} + \frac{ n^{1/4} }{ \sqrt{p} } )$;
\item $\max\limits_{k\in[K]}n^{-1}\sum_{t=1}^{n}| (\hat{\bF}\bH_0)_{jk} -f_{tk}|^2=O_{\mathrm{P}}(\frac{1}{n}+\frac{1}{p})$
\item $\max\limits_{ j \in [p-1] }n^{-1}\sum_{t=1}^{n}|\hat{u}_{ji}-u_{ji}|^2=O_{\mathrm{P}}(\frac{\log p}{n}+\frac{1}{p})$;
\item $\|\hat{\bU}-\bU\|_{\max}=o_{\mathrm{P}}(1)$.
\end{enumerate}
\end{lem}

A practical issue arises on how to choose the number of factors.  We adapt the ratio method for the numerical studies in this paper, as it involves only one tuning parameter.  Let $\lambda_k(\bX\bX^{T})$ be the
$k$th largest eigenvalue of $\bX\bX^{T}$ and  $K_{max}$ be a prescribed upper bound. The number of factors can be consistently estimated by \citep{Luo09,LamYao,Ahn_Horen_2013}
\begin{equation}
\hat{K}=\argmax\limits_{ k \leq K_{max}} \frac{\lambda_k(\bX\bX^{T})}{\lambda_{k+1}(\bX\bX^{T})}.
\end{equation}
Other viable method includes the information criteria in \cite{Bai_Ng_02}.

\subsection{Decorrelated variable screening}\label{sec_screening}
Screening methods \citep[e.g.][]{Fan_Lv_2008, Fan_Song_09, WLe2016}  are computationally attractive and thus  popular for ultra-high dimensional data analysis. However, the screening methods tend to include too many variables when there exist strong correlations among covariates \citep{Fan_Lv_2008, WLe2016}.  As an extension of FarmSelect, we propose the following conditional variable screening method to tackle this problem.

\begin{description}
\item [Step 1:] {\it Initial estimation}. We fit the approximate factor model (\ref{factor-x}) to obtain $\hat{\bB}$, $\hat{\bF}$, $\hat{\bU}=\bX-\hat{\bF}\hat{\bB}^{T}$ and $\hat\bU_1=(\mathbf{1}_n,\hat\bU)$.

\item [Step 2:] {\it Augmented marginal regression}. For $j\in[p]$, let $\hat{\bv}_j$  be the $j$-th column of $\hat\bU_1$ and
\begin{equation}\label{L1-ppmle-screening}
\hat{\theta}_j =\argmin\limits_{\bgamma\in \R^{K},\theta\in\R}
L_n (\by, \hat{\bF} \bgamma+\hat{\bv}_j^T\theta).
\end{equation}

\item [Step 3] {\it Screening}. Sort the $\{\hat{\theta}_j\}_{j=1}^p$ in terms of their absolute values, and take the largest ones.
\end{description}
For sparse linear regression, our screening method reduces to the factor-profiled screening proposed by \cite{Wang2012}.

\section{Theoretical results}\label{theory}
\subsection{General results}\label{sec_general}
We first present general model selection results for the FarmSelect estimator (\ref{L1-ppmle}). Without loss of generality, we assume the last $K$ variables are not penalized. Let $L_n:\mathbb{R}^{p+K}\rightarrow\mathbb{R}$ be a convex loss function,
$\btheta^*\in\R^{p+K}$ and  $\bbeta^* = \btheta^*_{[p]}$ be the sparse sub-vector of interest. Then $\btheta^*$ and $\bbeta^*$ are estimated via
$$
\hat{\btheta}=\argmin_{\btheta \in \R^{p+K} }\{L_n (\btheta )+\lambda\|\btheta_{[p]}\|_1 \} \quad \mbox{and} \quad \hat{\bbeta}=\hat{\btheta}_{[p]},
$$
respectively. Further, denote $S = \mbox{supp}(\btheta^*)$, $S_1 = \mbox{supp}(\bbeta^*)$  and $S_2=[p+K]\backslash S$.

\medskip

\begin{ass}[Smoothness]\label{assump-smoothness}
$L_n(\btheta)\in C^2(\R^{p+K})$ and there exist $A>0$, $M>0$ such that $\|\nabla_{\cdot S}^2L(\btheta)-\nabla_{\cdot S}^2 L_n(\btheta^*)\|_{\infty}\leq M \|\btheta-\btheta^*\|_2$ whenever $\mathrm{supp}(\btheta)\subseteq S$ and $\|\btheta-\btheta^*\|_2\leq A$.
\end{ass}
\begin{ass}[Restricted strong convexity]\label{assump-RSC}
There exist $\kappa_2>\kappa_{\infty}>0$ such that $\|[ \nabla^2_{SS}L_n(\btheta^*) ]^{-1}\|_{\infty}\leq\frac{1}{2\kappa_{\infty}}$ and $\| [ \nabla^2_{SS}L_n(\btheta^*) ]^{-1}\|_{2}\leq\frac{1}{2\kappa_2}$.
\end{ass}
\begin{ass}[Irrepresentable condition]\label{assump-IC}
$\|\nabla^2_{S_2 S} L_n(\btheta^{\ast}) [ \nabla_{S S}^2 L_n(\btheta^{\ast}) ]^{-1}\|_{\infty}\leq 1-\tau$ for some $\tau\in(0,1)$.
\end{ass}

Assumptions \ref{assump-smoothness} -- \ref{assump-IC} are standard in the studies of high-dimensional regularized estimators  \citep[e.g.][]{ Negahban_12, Lee15}. Based on them, we introduce the following theorem of $L^p$ ($p=1,2,\infty$) error bounds and sign consistency for the FarmSelect estimator.

\medskip

\begin{thm}\label{consistency-general}
(i) \textbf{\textit{Error bounds}} :
	Under Assumptions \ref{assump-smoothness} -- \ref{assump-IC}, if
	\begin{equation}
	\frac{7}{\tau}\|\nabla L_n(\btheta^*)\|_{\infty}<\lambda<\frac{\kappa_2}{4\sqrt{|S|}}\min\left\{A,\frac{\kappa_{\infty}\tau}{3M}\right\},
	\end{equation}
	then $\supp(\hat{\btheta})\subseteq S$ and
	\begin{equation*}
	\begin{split}
	&\|\hat{\btheta}-\btheta^{\ast}\|_{\infty}\leq \frac{3}{5\kappa_{\infty}}(\|\nabla_S L_n(\btheta^*)\|_{\infty}+\lambda),\\
	&\|\hat{\btheta}-\btheta^{\ast}\|_2\leq \frac{2}{\kappa_2}(\|\nabla_S L_n(\btheta^*)\|_2+\lambda\sqrt{|S_1|}),\\
	&\|\hat{\btheta}-\btheta^{\ast}\|_1\leq \min\Big\{\frac{3}{5\kappa_{\infty}}(\|\nabla_S L_n(\btheta^*)\|_1+\lambda|S_1|),
	\frac{2\sqrt{|S|}}{\kappa_2}(\|\nabla_S L_n(\btheta^*)\|_2+\lambda\sqrt{|S_1|})\Big\}
	.\\
	\end{split}
	\end{equation*}
(ii) \textbf{\textit{Sign consistency}} :	In addition, if the following two conditions
	\begin{equation}
	\begin{split}
	&\min\{|\bbeta^*_j|:\bbeta^*_j\neq 0,~j\in[p]\}> \frac{C}{\kappa_{\infty}\tau}\|\nabla L_n(\btheta^*)\|_{\infty},\\
	&\|\nabla L_n(\btheta^*)\|_{\infty}<\frac{\kappa_2\tau}{7C\sqrt{|S|}}\min\left\{A,\frac{\kappa_{\infty}\tau}{3M}\right\}
	\end{split}
	\end{equation}
	hold for some $C\geq 5$, then by taking $\lambda\in (\frac{7}{\tau}\|\nabla L_n(\btheta^*)\|_{\infty},\frac{1}{\tau}(\frac{5C}{3}-1)\|\nabla L_n(\btheta^*)\|_{\infty})$, the estimator achieves the sign consistency $\sign(\hat{\bbeta})=\sign(\bbeta^{\ast})$.
\end{thm}

\medskip

\begin{rem}\label{rem_4_1}  Theorem \ref{consistency-general} shows how the correlated covariates  affect the sign consistency as well as  error bounds. To achieve the sign consistency, the tuning parameter $\lambda$ should scale with $\tau^{-1}$. Therefore, the $L^{\infty}$ and $L^2$ errors will scale with $(\kappa_{\infty}\tau)^{-1}$ and $(\kappa_2\tau)^{-1}$, respectively.  When the covariates are highly correlated, the {\it irrepresentable condition} will get violated or the parameter $\tau\in(0,1)$ is very small.
As a result,  the model selection procedures will fail to achieve the sign consistency and the error bounds will be suboptimal. On the other hand, the optimal error bounds require a  small $\lambda$, which typically leads to an overfitted model. One can see a trade-off between model selection and parameter estimation due to the existence of dependency.
\end{rem}

\begin{rem} The $L^1$ and $L^2$ error bounds in Theorem \ref{consistency-general} depend on $|S_1|$, the number of active variables. They stem from the bias induced by the penalty. To reduce the bias, it is desirable to penalize as few active variables as possible. This phenomena motivates  FarmSelect to adopt a penalized profile likelihood form by not imposing penalty on the nuisance parameter $\bgamma$.
\end{rem}

As discussed in Remark \ref{rem_4_1}, when the covariates are highly correlated, the {\it irrepresentable condition} may not hold, or has a very small $\tau$. This makes the model selection consistency either very hard to achieve or incompatible with low estimation error bounds. Therefore, the FarmSelect strategy can improve the model selection consistency and reduce the estimation error bounds if $\bX$ can be decomposed into $\bF\bB^T+\bU$ such that $(\bU,\bF)$ is well-behaved. This is due to the fact that the {\it irrepresentable condition} is easier to hold with positive $\tau$ bounded away from zero after the decorrelation step. To this end, any effective  decorrelation procedure can be incorporated into this frame work.

\subsection{FarmSelect with approximate factor model}
Now we study the FarmSelect estimator when the covariates $\bX$ admits the approximate factor model (\ref{factor-x}). The oracle procedure uses true augmented covariates $\bw_t=(1,\bu_t^T,\bff_t^T)^T$ for $t\in[n]$ and solves
\begin{equation*}
\min_{\btheta}\{L_n(\by,\bW \btheta)+\lambda \|\btheta_{[p]}\|_{1}\},
\end{equation*}
where $\bW=(\bw_1^T,\cdots,\bw_n^T)^T=(\bU_1,\bF)$.
However,  $\bW$ is not observable in practice. Hence we need to use its estimator $\hat{\bW}$ and solve
\begin{equation*}
\min_{\btheta}\{L_n(\by,\hat{\bW}\btheta)+\lambda \|\btheta_{[p]}\|_{1}\}.
\end{equation*}
Below the error induced by the factor estimation will be studied carefully. To deliver a clear discussion on the conditions and results, we focus on the FarmSelect estimator for the generalized linear model (\ref{traditional-L1-PMLE}), and assume that the covariates are generated from the approximate factor model (\ref{factor-x}).

\medskip
\begin{ass}[Smoothness]\label{assump-regularity-domain}
	$b(z)\in C^3(\R)$. For some constants $M_2$ and $M_3$, we have $0\leq b''(z)\leq M_2$ and $|b'''(z)|\leq M_3$, $\forall z$.
\end{ass}


\begin{ass}[Restricted strong convexity and irrepresentable condition]\label{assump-RSC-1}
Let $\btheta^*=\begin{pmatrix}
\bbeta^*\\
\bB_0^T \bbeta^*
\end{pmatrix}$.
Assume the existence of $\kappa_2>\kappa_{\infty}>0$ and $\tau\in(0,1)$ such that
	\begin{equation}\label{eq_ass_45}
	\begin{split}
	&\|[ \nabla^2_{SS}L_n(\by,\bW\btheta^*) ]^{-1}\|_\ell \leq\frac{1}{4\kappa_\ell}, \quad \mbox{for } \ell = 2 \mbox{ and } \infty, \\ 
	&\|\nabla^2_{S_2 S} L_n(\by,\bW\btheta^{\ast}) [ \nabla_{S S}^2 L_n(\by,\bW\btheta^{\ast}) ]^{-1}\|_{\infty}\leq 1-2\tau.
	\end{split}
	\end{equation}
\end{ass}


\begin{ass}[Estimation of factor model]\label{assump-factor}
	$\|\bW\|_{\max}\leq \frac{M_0}{2}$ for some constant $M_0>0$. In addition, there exist $K\times K$ nonsingular matrix $\bH_0$,  and $\bH=\begin{pmatrix}
\bI_p & \mathbf{0}_{p\times K}\\
	\mathbf{0}_{K \times p } & \bH_0
	\end{pmatrix}$ such that for $\overline{\bW}=\hat{\bW}\bH$, we have $\|\overline{\bW}-\bW\|_{\max}\leq\frac{M_0}{2}$ and $\max_{j\in[p+K]}\Big(\frac{1}{n}\sum_{t=1}^{n}|\overline{w}_{tj}-w_{tj}|^2\Big)^{1/2}\leq\frac{2\kappa_{\infty}\tau}{3M_0M_2|S|}$.
\end{ass}

\medskip

\begin{rem}
(i) Assumption \ref{assump-regularity-domain} holds for a large family of generalized linear models. For example, linear model has $b(z)=\frac{1}{2}z^2$, $M_2=1$ and $M_3=0$; logistic model has $b(z)=\log(1+e^z)$ and finite $M_2$, $M_3$.
(ii) Note that the first inequality in (\ref{eq_ass_45}) involves only a small matrix and holds easily, and the second inequality there is related to the generalized irrespresentable condition.
Standard concentration inequalities (e.g. the Bernstein inequality for weakly dependent variables in \cite{MPR11}) yield that Assumption \ref{assump-RSC-1} holds with high probability as long as $\E [\nabla^2 L_n(\by,\bW\btheta^*)]$ satisfies similar conditions.
(iii) Under the conditions of Lemma \ref{lem-factor}, we have $\| \overline\bW - \bW \|_{\max} =o_{\mathrm{P}} (1)$ and
$\max_{ j\in[p+K]}\Big(\frac{1}{n}\sum_{t=1}^{n}|\overline w_{tj}-w_{tj}|^2\Big)^{1/2}=O_{\mathrm{P}}(\sqrt{\frac{\log p}{n}} + \frac{1}{\sqrt{p}})$,
where $\overline{\bW}=\hat\bW \bH$,
$\bH=\begin{pmatrix}
\bI_p & \mathbf{0}_{p\times K}\\
\mathbf{0}_{K\times p} & \bH_0
\end{pmatrix}$ and some proper $\bH_0$.
Hence $|S|^2 ( \frac{\log p}{n} + \frac{1}{p})=O(1)$ can guarantee Assumption \ref{assump-factor} to hold with high probability.

\end{rem}

\medskip

\begin{thm}\label{consistency-estimated-factors}
	Suppose Assumptions \ref{assump-regularity-domain}-\ref{assump-factor} hold. Define
	$M=M_0^3M_3|S|^{3/2}$ and
	\begin{equation*}
	\varepsilon=\max_{j\in[p+K]}\Big|\frac{1}{n}\sum_{t=1}^{n}\overline{w}_{tj}[-y_t+b'( (1,\bx_t^T) \bbeta^*)]\Big|.
	\end{equation*}
	If $\frac{7\varepsilon}{\tau}<\lambda<\frac{\kappa_2\kappa_{\infty}\tau}{12M\sqrt{|S|}}$,
	then we have $\supp(\hat{\bbeta})\subseteq \supp(\bbeta^*)$ and
\begin{eqnarray*}
 \|\hat{\bbeta}-\bbeta^{\ast}\|_{\infty}\leq \frac{6\lambda}{5\kappa_{\infty}}, \qquad  \|\hat{\bbeta}-\bbeta^{\ast}\|_2\leq \frac{4\lambda\sqrt{|S|}}{\kappa_2},
\qquad \|\hat{\bbeta}-\bbeta^{\ast}\|_1\leq \frac{6\lambda|S|}{5\kappa_{\infty}}.
\end{eqnarray*}
In addition, if $\varepsilon<\frac{\kappa_2\kappa_{\infty}\tau^2}{12CM\sqrt{|S|}}$ and
$\min\{|\bbeta^*_j|:\bbeta^*_j\neq 0,~j\in[p]\}> \frac{6C\varepsilon}{5\kappa_{\infty}\tau}$
	hold for some $C>7$, then by taking $\lambda\in (\frac{7}{\tau}\varepsilon,\frac{C}{\tau}\varepsilon)$ we can achieve the sign consistency $\sign(\hat{\bbeta})=\sign(\bbeta^{\ast})$.	
\end{thm}

By taking $\lambda\asymp \varepsilon$, one can achieve the sign consistency and  $\|\hat{\bbeta}-\bbeta^*\|_{\infty}/\varepsilon=O_{\mathrm{P}}(1)$, $\|\hat{\bbeta}-\bbeta^*\|_2/\varepsilon=O_{\mathrm{P}}(\sqrt{|S|})$ and $\|\hat{\bbeta}-\bbeta^*\|_1/\varepsilon=O_{\mathrm{P}}(|S|)$. Hence $\varepsilon$ is a key quantity characterizing the error rate of our FarmSelect estimator, whose size is controlled using the following lemma.

\begin{lem}\label{lem-epsilon}
Let $\eta_t = y_t - b'( (1,\bx_t^T) \bbeta^*)$ and $\bw_t = ( 1,\bu_t^T,\bff_t^T)^T$. Assume that $\{ \bw_t, \eta_t \}_{-\infty}^{\infty}$ is strictly stationary and satisfies the following conditions
\begin{enumerate}
\item $\E ( \eta_t \bw_t )=\mathbf{0}$;
\item There exist constants $b,\gamma_1>0$ such that $\mathrm{P}( | \eta_{t} |> s  ) \leq \exp(1-(s/b)^{\gamma_1})$ for all $t\in \mathbb{Z}$ and $s\geq 0$;
\item There existconstants $c,\gamma_3>0$ such that for all $T\geq 1$,
\[
\sup_{
	A\in\overline{\mathcal{F}}_{-\infty}^{0},B\in \overline{\mathcal{F}}_{T}^{\infty}
} |\mathrm{P}(A)\mathrm{P}(B)-\mathrm{P}(AB)|
\leq \exp(-cT^{\gamma_3}),
\]
where $\overline{\mathcal{F}}_{-\infty}^{0}$ and $\overline{\mathcal{F}}_{-\infty}^{0}$ denote the $\sigma$-algebras generated by $\{ (\bw_t,\eta_t): i\leq 0\}$ and $\{ (\bw_t,\eta_t): i\geq T\}$ respectively;
\end{enumerate}
In addition, suppose that the assumptions in Lemma \ref{lem-factor} hold. Then we have
\[
\varepsilon=\max_{j\in[p+K]} \left|\frac{1}{n}\sum_{t=1}^{n}\overline{w}_{tj} \eta_t \right|=O_{\mathrm{P}}(\sqrt{\frac{\log p}{n}} + \frac{1}{\sqrt{p}}).
\]
\end{lem}
Recall that the assumption $\E ( \eta_t \bw_t )=\mathbf{0}$ has been used in Lemma \ref{lem-identifiability} as a cornerstone of our FarmSelect methodology. The rest in the list are standard conditions similar to Assumptions \ref{assump-factor-2}-\ref{assump-factor-4}. All of them are mild and interpretable.

Lemma \ref{lem-epsilon} asserts that $\varepsilon=O_{\mathrm{P}}(\sqrt{\frac{\log p}{n}} + \frac{1}{\sqrt{p}})$. The first term $\sqrt{\frac{\log p}{n}}$ corresponds to the optimal rate of convergence for high-dimensional $M$-estimator \citep[e.g.][]{Bickel_09}. The second term $\frac{1}{\sqrt{p}}$ is the price we pay for factor estimation, which is negligible if $n=O(p\log p)$. In that high-dimensional regime, all the error bounds for $\| \hat\bbeta - \bbeta^* \|_{\ell}$ ($\ell=1,2,\infty$) match the optimal ones in the literature.

\medskip

\section{Simulation study}

\subsection{Example 1: Linear regression}

We study a simulated example for high dimensional sparse linear regression with correlated covariates. The correlation structure is calibrated from S\&P 500 monthly excess returns between 1980 and 2012.
 Throughout the numerical studies of this paper, the tuning parameter $\lambda$ is selected by the 10-fold cross validation. The model selection performance is measured by the model selection consistency rate and the sure screening rate. The former is the proportion of simulations that the selected model is identical to the true one and the latter is the proportion of simulations that the selected model contains all important variables.

\subsubsection*{Calibration and data generation process}\label{sim_cali}
We calculate the centered monthly excess returns for the stocks in S\&P 500 index that have complete records from January 1980 to December 2012. The data, collected from
CRSP\footnote{Center for Research in Security Prices Database, see \url{http://www.crsp.com/} for more details.}
, contains the returns of 202 stocks with a time span of 396 months. Denote the centred monthly excess returns as ${\bz}_t$, $t=1, \ldots, 396$. The calibration and data generation procedure are outlined as follows.
\begin{itemize}
\item[(1)] Fit ${\bz}_t$ with a three factor model. We apply PCA on the sample covariance of $\{\bz_t\}_{t=1}^{396}$  and denote $\lambda_k$ and $\bxi_k$, $k=1,2,3$, as the top three eigenvalues and corresponding eigenvectors. We estimate loadings $\tilde{\bB}=(\sqrt{\lambda_1}\bxi_1, \sqrt{\lambda_2}\bxi_2, \sqrt{\lambda_3}\bxi_3)$
    and $\tilde{\bff}_t = (\lambda_1^{-1/2} \bxi_1^T \bz_t, \lambda_2^{-1/2} \bxi_2^T \bz_t, \lambda_3^{-1/2} \bxi_1^T \bz_t)^T$.

\item[(2)] Calculate $\bSigma_{\bB}$ as the sample covariance of the rows of $\tilde{\bB}$, which is $\diag(\lambda_1, \lambda_2, \lambda_3)$. Generate loading matrix $\bB$ whose rows are draws from i.i.d. $N({\bf 0}, \ \bSigma_{\bB})$.

\item[(3)] Fit VAR(1) model $\tilde{\bff}_t=\bPhi \tilde{\bff}_{t-1} + \bfeta_t$. Denote $\hat{\bPhi}$ the estimate of $\bPhi$ and calculate $\bSigma_{\bfeta}={\bf I}-\hat{\bPhi}\hat{\bPhi}^{T}$. Generate $\bff_t$ from the VAR(1) model $\bff_t=\hat{\bPhi}\bff_{t-1}+\bfeta_t$ with $\bff_0={\bf 0}$, where $\bfeta_t$ is generated from i.i.d. $N({\bf 0}, \ \bSigma_{\bfeta})$.

\item[(4)]  Calculate the residual $\tilde{\bu}_t= \bz_t-\tilde{\bB}\tilde{\bff}_t$ and $\bSigma_{\bu}$ the sample covariance matrix of $\tilde{\bu}_t$. Denote $\sigma_u^2$ the average of the diagonal entries of $\bSigma_{\bu}$. Generate covariates $\bx_t$ from a factor model $\bx_t=\bB \bff_t + \bu_t$ where the entries in $\bu_t$ are drawn from i.i.d. $N(0, \ \sigma_u^2)$.

\item[(5)]Generate the response $y_t$ from a sparse linear model $y_t=\bx_t^{T} \bbeta^* + \eps_t$. The true coefficients are
$\bbeta^{\ast}=(\beta_1, \cdots , \beta_{10}, {\bf 0}^{T}_{(p-10)})^{T}$,   and the nonzero coefficients $\beta_1, \cdots , \beta_{10}$ are drawn from i.i.d. Uniform(2,5). We draw $\eps_t$ from an AR(1) model $\eps_t=0.5\eps_{t-1}+\gamma_t$ with $\gamma_t \sim N(0, \ 0.3)$.

\end{itemize}

The results of the calibrated parameters are presented in Table \ref{Tab_1}.
\begin{table}[htbp]
\begin{center}
\caption{Parameters calibrated from S\&P 500 returns}
\vspace*{0.1in}
\label{Tab_1}
\begin{tabular}{ccc|ccc|ccc|c}
\hline\hline
\multicolumn{3}{c|}{  $\bSigma_{\bB}$}  & \multicolumn{3}{c|}{  $\hat{\bPhi}$}
& \multicolumn{3}{c|}{ $\bSigma_{\bfeta}$}
&  $\sigma_u^2$
\\\hline
0.5237 &0 &0 &0.1897  &-0.0375 &-0.0223 &0.9621  &-0.0056 &0.0182  &
\\
0 &0.2884 &0 &0.0630  &0.1553  &0.0206  &-0.0056 &0.9715  &-0.0078 & 0.0146
\\
0 &0 &0.2372 &-0.0432 &0.0102  &0.4343  &0.0182  &-0.0078 &0.8094  &
\\\hline\hline
\end{tabular}
\end{center}
\end{table}

\subsubsection*{Impacts of Irrepresentable Condition}
First, we show LASSO performs poorly in terms of model selection consistency rate when the {\it irrepresentable condition} is violated, while FarmSelect can consistently select the correct model.
Let $n=100$ and $p=500$. Denote by $\Gamma_\infty=\|\bX_{S^c}^{T}\bX_{S}(\bX_{S}^T\bX_{S})^{-1}\|_{\infty}$. When $\Gamma_\infty<1$ the {\it irrepresentable condition} holds and otherwise it is violated. We simulate 10,000 replications.  For each replication, we calculate $\Gamma_\infty$ and apply both LASSO and FarmSelect for model selection. Then we calculate the model selection consistency rate within each small interval around $\Gamma_\infty$ (a nonparametric smoothing).
 The results are presented in Figure \ref{Fig_1}. According to Figure \ref{Fig_1}, both FarmSelect and LASSO have high model selection consistency rate when
$\Gamma_\infty < 1$. This shows FarmSelect does not pay any price under the weak correlation scenario.
As $\Gamma_\infty$ grows beyond 1, the correct model selection rate of LASSO drops quickly.
When the {\it irrepresentable condition} is strongly violated (e.g. $\Gamma_\infty > 1.5$), the correct model selection rate of LASSO is close to zero. On the contrary, FarmSelect has  high selection consistency rates regardless of $\Gamma_\infty$.

\begin{figure}[htbp]
 \centering
 \includegraphics[width=5.5 in]{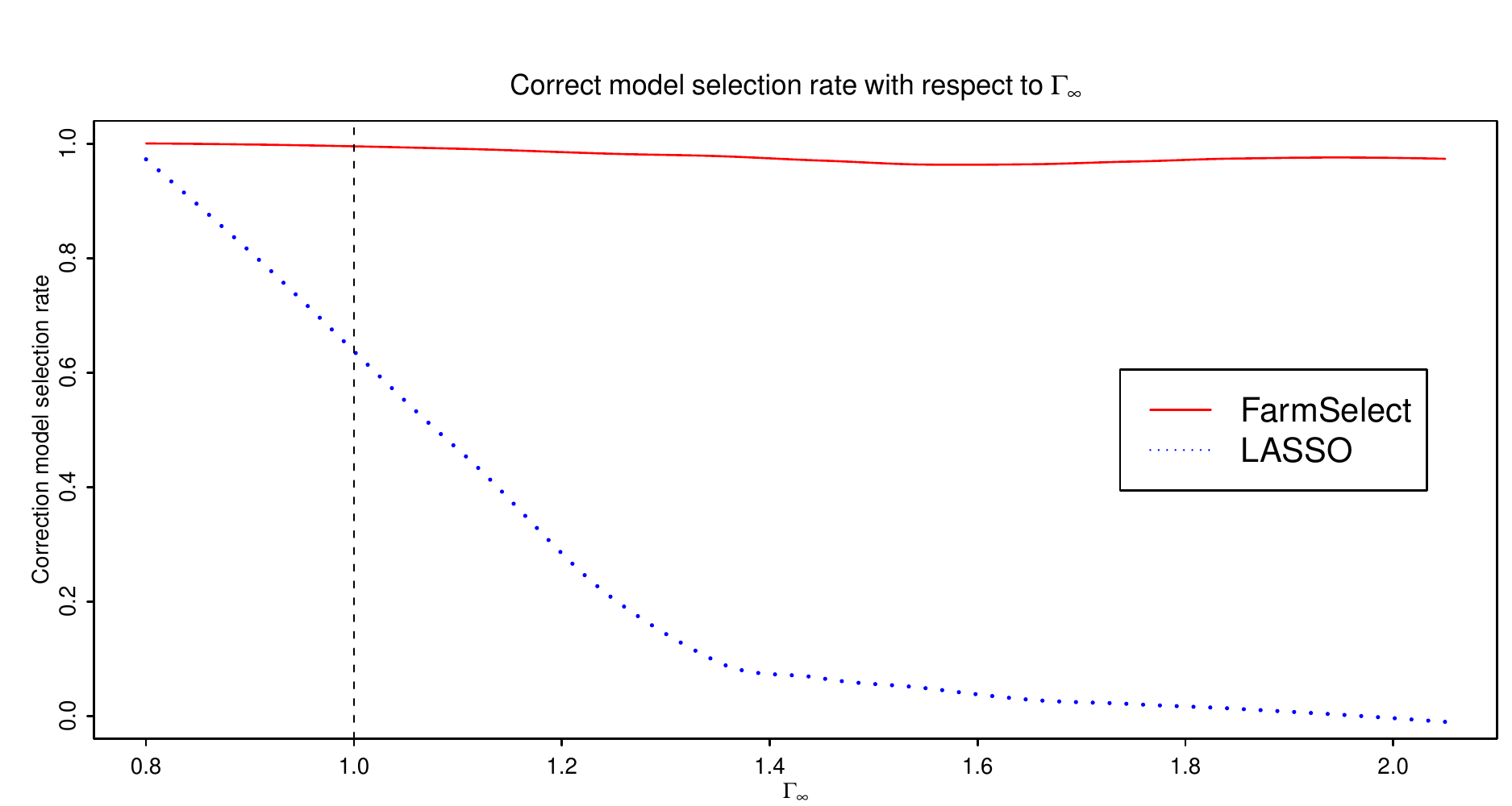}
  \caption{Relationship between model selection consistency rate and {\it irrepresentable condition}. Among the 10,000 replications, more than 9,500 replications have $\Gamma_\infty>1$ and more than 8,000 replications have $\Gamma_\infty>1.5$.  }
 \label{Fig_1}
\end{figure}

\subsubsection*{Impacts of sample size}
Second, we examine the model selection consistency with a fixed dimensionality and an increasing sample size.
We fix $p=500$ and let $n$ increase from 50 to 150. For each given sample size, we simulate 200 replications and calculate the model selection consistency rates and the sure screening rates for LASSO, SCAD, elastic net and FarmSelect, respectively. For the elastic net, we set $\lambda_1=\lambda_2\equiv\lambda$. The results are presented in Figure \ref{Fig_sim_1} (a) and Figure \ref{Fig_sim_1} (b). Figure \ref{Fig_sim_1} (a) shows that model selection consistency rates of LASSO, SCAD and elastic net do not enjoy fast convergence to 1 when the sample size increases, while the one of FarmSelect equals to one as long as the sample size exceeds 100. Similar phenomena are observed from sure screening rates.
To demonstrate the prediction performance, we report the mean estimation error
$\|\hat{\bbeta}-\bbeta^{\ast}\|_2$ for each method, which is a good indicator of the prediction error. The estimation errors are reported in Figure \ref{Fig_sim_1} (c). When the sample size is small, LASSO, SCAD and elastic net have large estimation errors since they tend to select overfitted models.

\subsubsection*{Impacts of dimensionality}
Third, we assess the model selection performance  when the dimensionality $p$ is growing beyond $n$ and diverging. We fix $n=100$ and let $p$ grow from 200 to 1000. For each given $p$, we simulate 200 replications and calculate the  model selection consistency rate of LASSO, SCAD, elastic net and FarmSelect respectively. The model selection consistency rates are presented in Figure \ref{Fig_sim_2}(a). According to Figure \ref{Fig_sim_2}(a),  the model selection consistency rate of FarmSelect stays close to 1 even as $p$ increases, whereas the rates for the other three methods  drop quickly. Again, we report the estimation errors in Figure \ref{Fig_sim_2}(b). As the dimensionality grows, FarmSelect has the least increase of estimation error.

\begin{figure}[htbp]
 \centering
 \includegraphics[width=5 in]{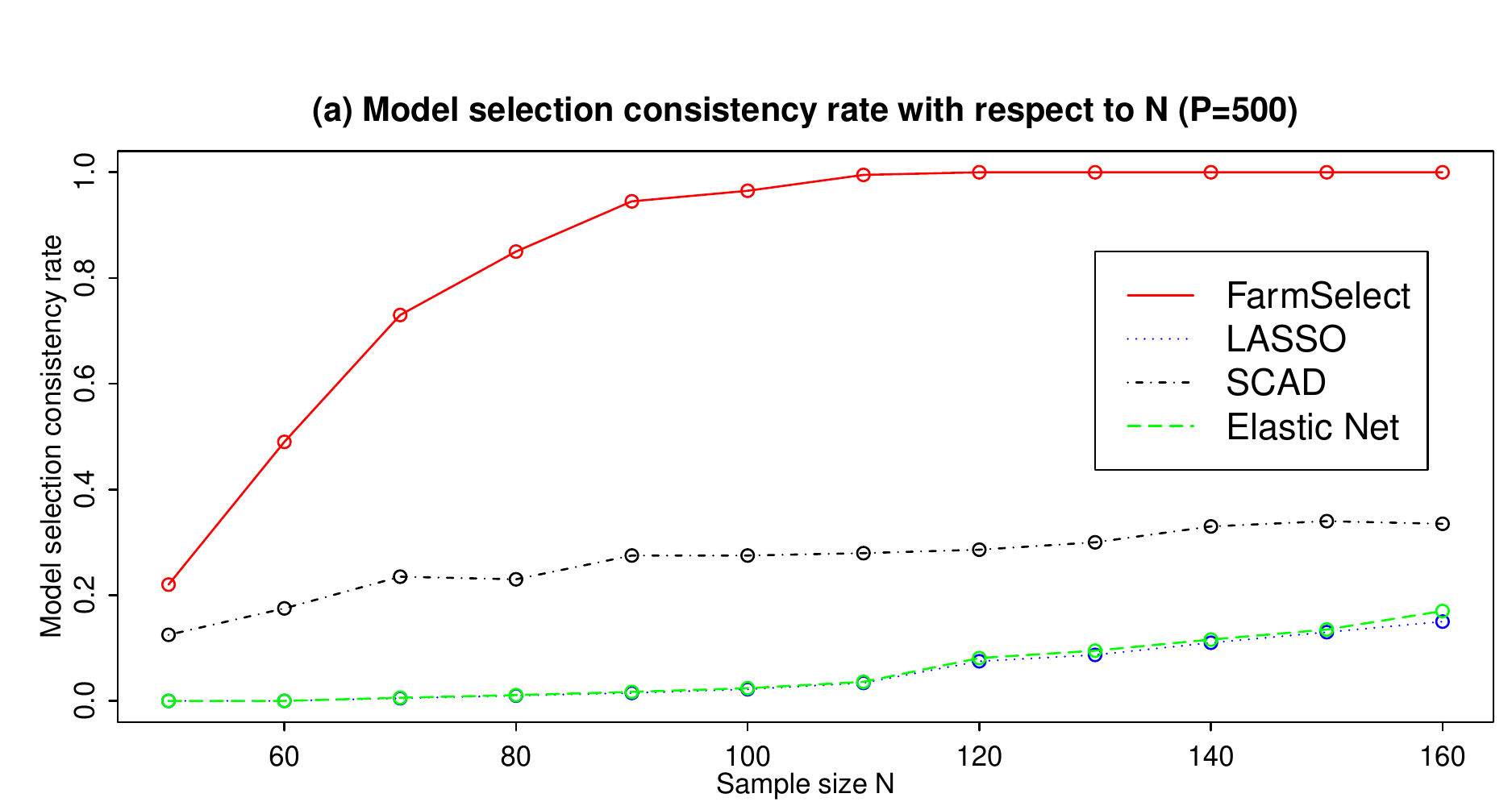}
 \includegraphics[width=5 in]{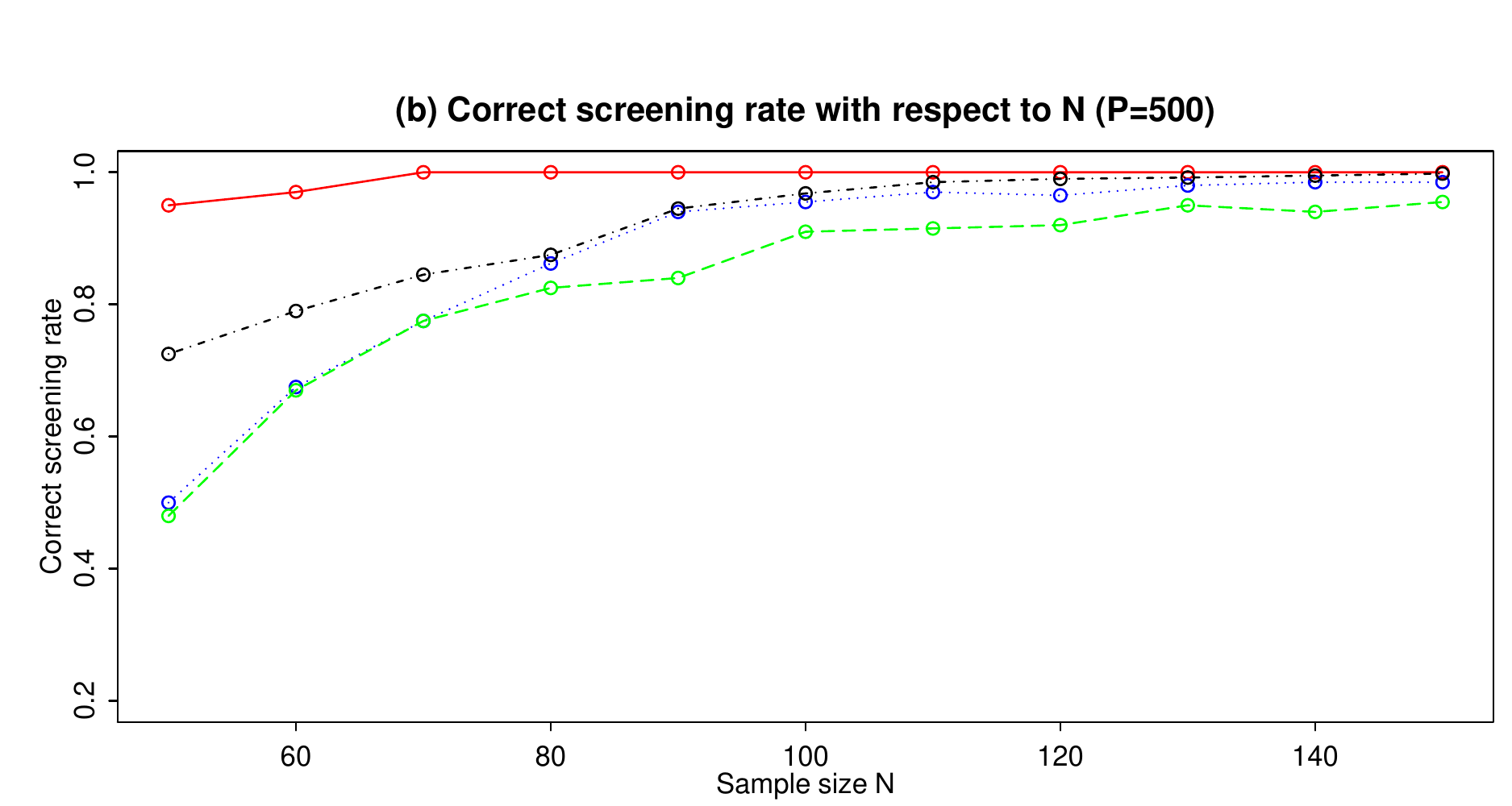}
 \includegraphics[width=5 in]{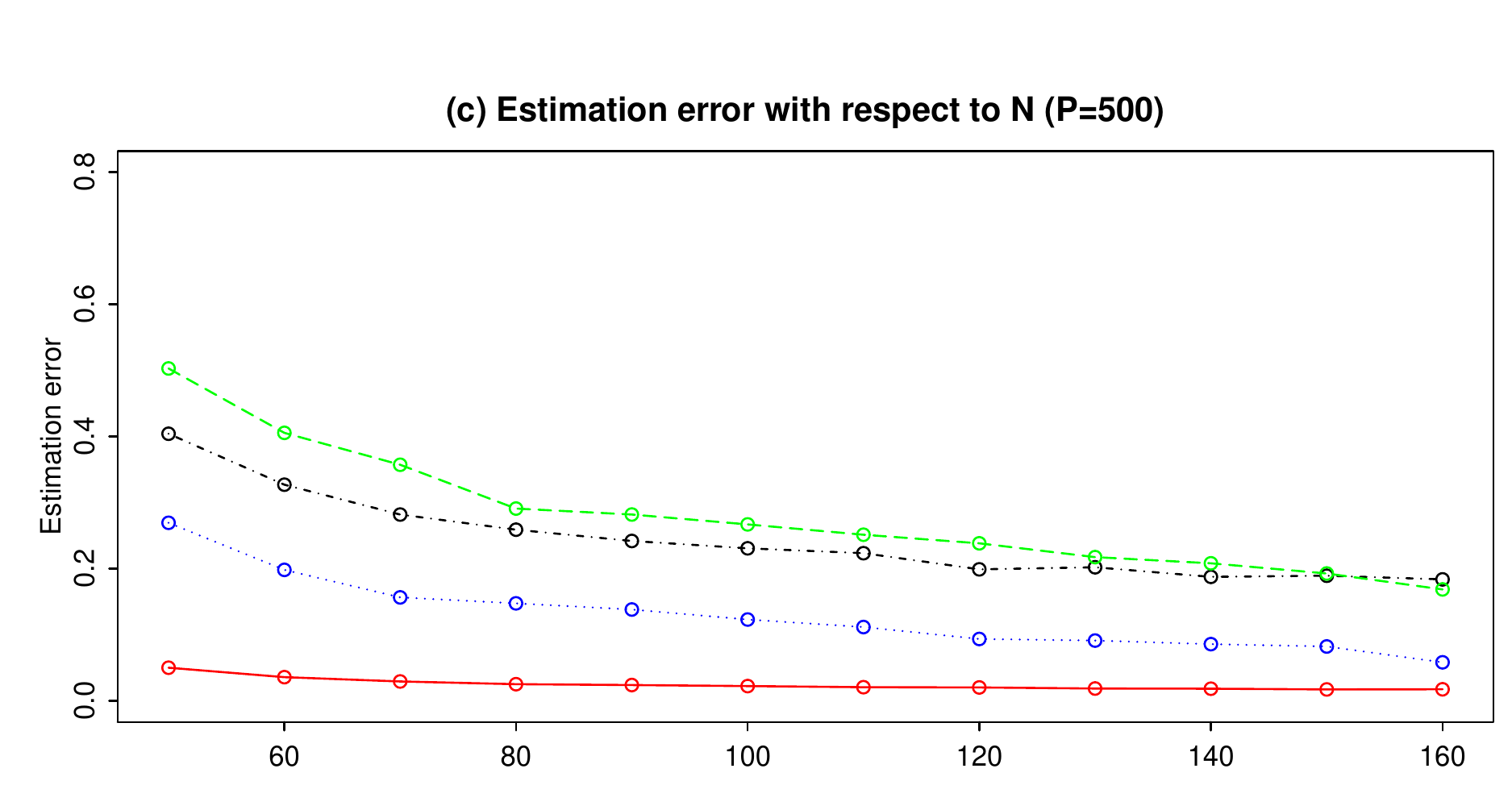}
  \caption{From above to below: (a) Model selection consistency rates with fixed $p$ and  increasing $n$;
   (b) Sure screening rates with fixed $p$ and  increasing $n$;
   (c) Mean estimation errors $\|\hat{\bbeta}-\bbeta^{\ast}\|_2$  with fixed $p$ and  increasing $n$.
    }
 \label{Fig_sim_1}
\end{figure}

\begin{figure}[htbp]
 \centering
 \includegraphics[width=5 in]{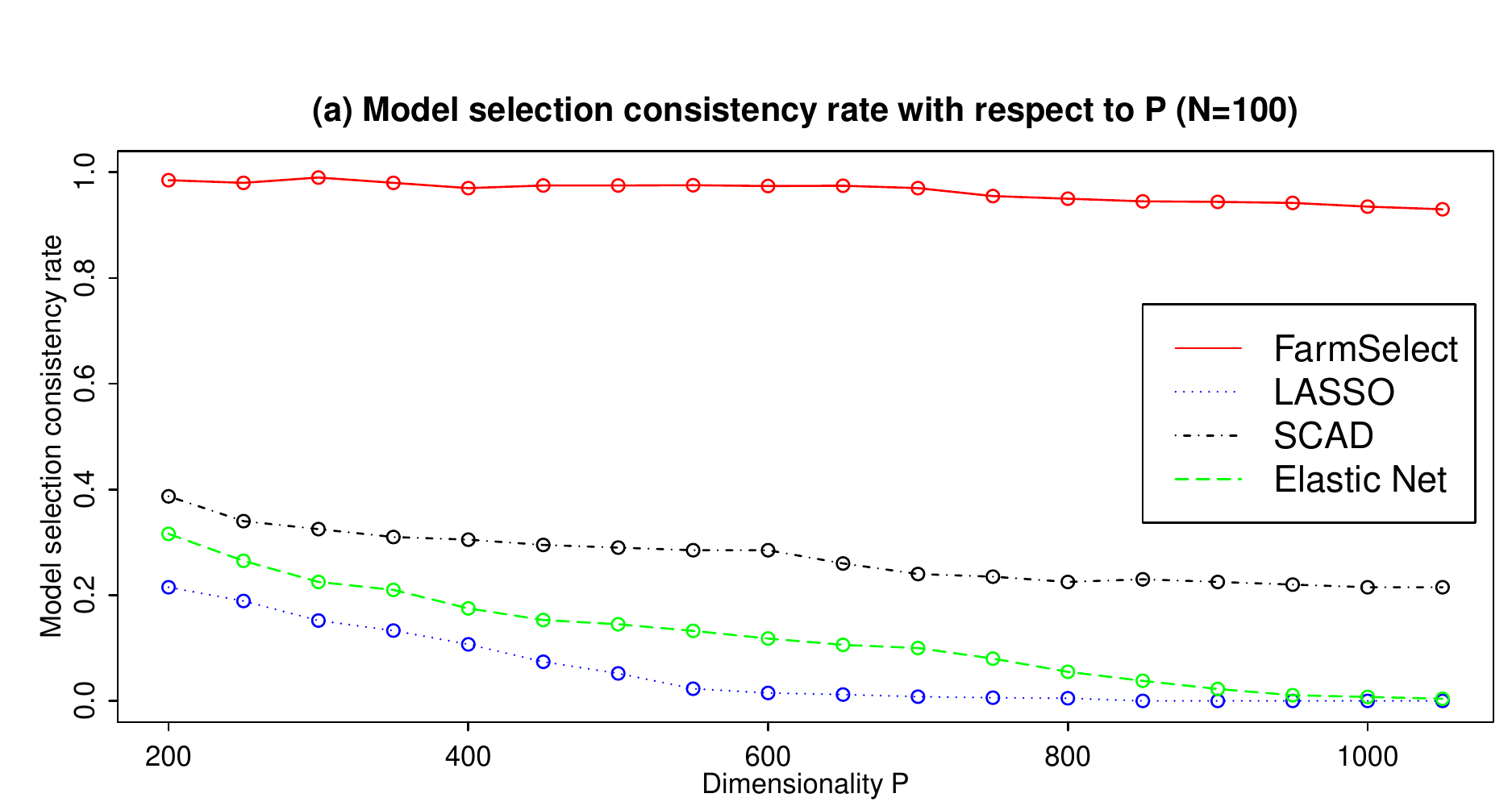}
 \includegraphics[width=5 in]{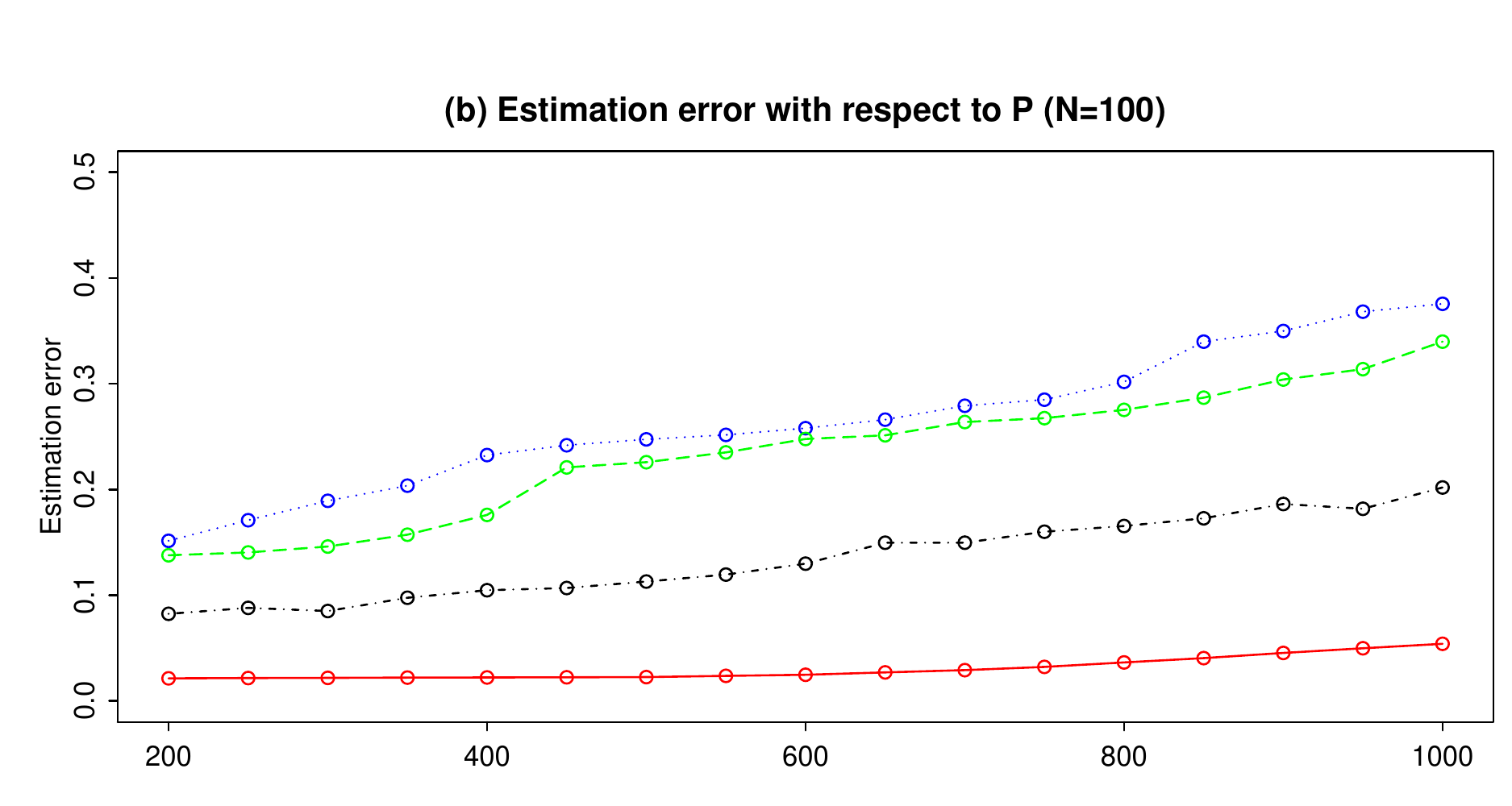}
  \caption{From above to below: (a) Model selection consistency rates with fixed $n$ and  increasing $p$; (b) Mean estimation errors $\|\hat{\bbeta}-\bbeta^{\ast}\|_2$ with fixed $n$ and  increasing $p$.
    }
 \label{Fig_sim_2}
\end{figure}

\subsection{Example 2: Logistic regression}

We consider the following logistic regression model whose conditional probability function is:
\begin{equation}\label{eq5.3}
\mathrm{P}(y_t = 1 | {\bf X}_t)=\frac{\exp({\bf X}_t^T\bbeta)}{1+\exp({\bf X}_t^T\bbeta)}, \quad
i=1, \ \cdots \ , N.
\end{equation}
We set sample size $n=200$ and  dimensionality $p=300,\ 400, \ 500$. The true coefficients are set to be
$\bbeta^{\ast}=(\bbeta_{(1)}^T, \ {\bf 0})^T$ with $\bbeta_{(1)}=(6,\ 5,\ 4)^T$. Hence  the true model size is 3.

The  covariates $\bX$ are generated from one of the following three models:
\begin{enumerate}
\item[(1)] Factor model $\bx_t=\bB \bff_t + \bu_t$ with $K=3$. Factors are generated from a stationary VAR(1) model $\bff_t=\bPhi \bff_{t-1} + \bfeta_t$ with $\bff_0={\bf 0}$. The $(i, j)$th entry of $\bPhi$ is set to be 0.5 when $i = j$ and $0.3^{|i-j|}$ when $i \neq j$.  We draw $\bB$, $\bu_t$ and $\bfeta_t$ from the i.i.d. standard Normal distribution.

\item[(2)] Equal correlated case. We draw $\bx_t$ from i.i.d. $N_p({\bf 0}, \bSigma)$, where $\bSigma$ has diagonal elements 1 and off-diagonal elements 0.8.

\item[(3)] Uncorrelated case. We draw $\bx_t$ from i.i.d. $N_p({\bf 0}, \bI)$

\end{enumerate}

We compare the model selection performance of FarmSelect with LASSO and simulate 100 replications for each scenario. The model selection performance is measured by  selection consistency rate,  sure screening rate and the average size of selected model. The results are presented in Table \ref{Tab_2} below. According to Table \ref{Tab_2}, FarmSelect pays no price for the uncorrelated case and outperforms LASSO for highly correlated cases.

\begin{table}[htbp]
\begin{center}
\caption{Model selection results of logistic regression ($n=200$)}
\label{Tab_2}
\begin{tabular}{c|ccc|ccc}
\hline\hline
&  \multicolumn{6}{c}{Factor model with $K=3$}
\\
& \multicolumn{3}{c}{FarmSelect} & \multicolumn{3}{c}{LASSO}
\\\cline{2-7}
   &  {\scriptsize Selection rate}  &  {\scriptsize Screening rate}
&  {\scriptsize Average model size}
&  {\scriptsize Selection rate}  &  {\scriptsize Screening rate}
&  {\scriptsize Average model size }
\\
$p=300$    &  0.94   &  1.00   &  3.07  &    0.07   &  0.98   &  12.61
\\
$p=400$    &  0.90   &  0.99   &  3.12  &    0.05   &  0.95   &  12.94
\\
$p=500$    &  0.83   &  0.98   &  3.18  &    0.03   &  0.93   &  15.07
\\\hline
&  \multicolumn{6}{c}{Equal correlated case}
\\
& \multicolumn{3}{c}{FarmSelect} & \multicolumn{3}{c}{LASSO}
\\\cline{2-7}
   &  {\scriptsize Selection rate}  &  {\scriptsize Screening rate}
&  {\scriptsize Average model size}
&  {\scriptsize Selection rate}  &  {\scriptsize Screening rate}
&  {\scriptsize Average model size }
\\
$p=300$    &  0.93   &  1.00   &  3.09  &    0.07   &  0.85   &  9.90
\\
$p=400$    &  0.89   &  1.00   &  3.14  &    0.05   &  0.80   &  10.82
\\
$p=500$    &  0.85   &  0.99   &  3.19  &    0.02   &  0.69   &  11.79
\\\hline
&  \multicolumn{6}{c}{Uncorrelated case}
\\
& \multicolumn{3}{c}{FarmSelect} & \multicolumn{3}{c}{LASSO}
\\\cline{2-7}
    &  {\scriptsize Selection rate}  &  {\scriptsize Screening rate}
&  {\scriptsize Average model size}
&  {\scriptsize Selection rate}  &  {\scriptsize Screening rate}
&  {\scriptsize Average model size }
\\
$p=300$    &  0.97   &  1.00   &  3.03  &    0.95   &  1.00   &  3.14
\\
$p=400$    &  0.93   &  1.00   &  3.07  &    0.91   &  1.00   &  3.34
\\
$p=500$    &  0.91   &  1.00   &  3.10  &    0.89   &  1.00   &  3.42
\\\hline\hline
\end{tabular}
\end{center}
\end{table}

\section{Prediction of U.S. bond risk premia}

In this section, we predict U.S. bond risk premia with a large panel of macroecnomic variables. The response variable is the monthly data of U.S. bond risk premia with maturity of 2 to 5 years between January 1980 and December 2015 containing 432 data points. The bond risk premia is calculated as the one year return of an $n$ years maturity bond excessing the one year maturity bond yield as the risk-free rate. The covariates are 134 monthly U.S. macroeconomic variables in the FRED-MD database\footnote{The FRED-MD is a monthly economic database updated by the Federal Reserve Bank of St. Louis which is public available at \url{http://research.stlouisfed.org/econ/mccracken/sel/}.} \citep{MN2016}. The covariates in the FRED-MD dataset are strongly correlated and can be well explained by a few principal components. To see this, we apply principal component analysis to the covariates and draw the scree plot of the top 20 principal components in Figure \ref{Fig_scree_1}. The scree plot shows the first principal component solely explains more than 60\% of the total variance. In addition, the first 5 principal components together explain more than 90\% of the total variance.

We apply one month ahead rolling window prediction with a window size of 120 months. Within each window, we predict the U.S. bond risk premia by a high dimensonal linear regression model of dimensionality 134.
We compare the proposed FarmSelect method with LASSO in terms of model selection and prediction. In addition, we include the principal component regression (PCR) in the competition of prediction. The FarmSelect is implemented by the {\sl FarmSelect} R package with default settings. To be specific, the loss function is $L_1$, the number of factors is estimated by the eigen-ratio method and the regularized parameter is selected by multi-fold cross validation. The LASSO method is implemented by the {\sl glmnet} R package. The PCR method is implemented by the {\sl pls} package in R. The number of principal components is chosen as 8 which is suggested in \cite{ludvigson2009macro}.

The prediction performance is evaluated by the out-of-sample $R^2$ which is defined as
$$
R^2=1-\frac{\sum_{t=121}^{432}(y_t-\hat{y_t})^2 }{\sum_{t=121}^{432}(y_t-\bar{y_t})^2},
$$
where $y_t$ is the response variable realized at time $t$, $\hat{y_t}$ is the predicted $y_t$ by one of the three methods above using the previous 120 months data, and $\bar{y_t}$ is the sample mean of the previous 120 months responses $(y_{t-120}, \ldots, y_{t-1})$, which represents a naive predictor. For FarmSelect and LASSO, we also report the average selected model size for prediction at time $t \in \{121, \cdots, 432\}$. The out-of-sample $R^2$ and average selected model size are reported in Table \ref{tab_real_1}. The detailed prediction performances can be viewed in Figure \ref{fig_real_1}. The results in Table \ref{tab_real_1} show that FarmSelect selects parsimonious models and achieves the highest $R^2$'s in all scenarios. On the contrary, LASSO may select redundant models as it ignores the correlations among covariates. To see this,  we rank the covariates according to the selection frequency.  The top 10 selected covariates and their frequencies are listed in Table \ref{tab_real_2}. According to Table \ref{tab_real_2}, LASSO tends to select some highly correlated covariates simultaneously. For instance, LASSO includes both Housing Starts Northeast and New Private Housing Permits, Northeast (SAAR) due to strong correlation between them. In addition, both Switzerland/U.S. and Japan/U.S. exchange rates enter the solution path of LASSO early, which can be another evidence of overfitting.

\begin{figure}[htbp]
 \centering
 \includegraphics[width=5 in]{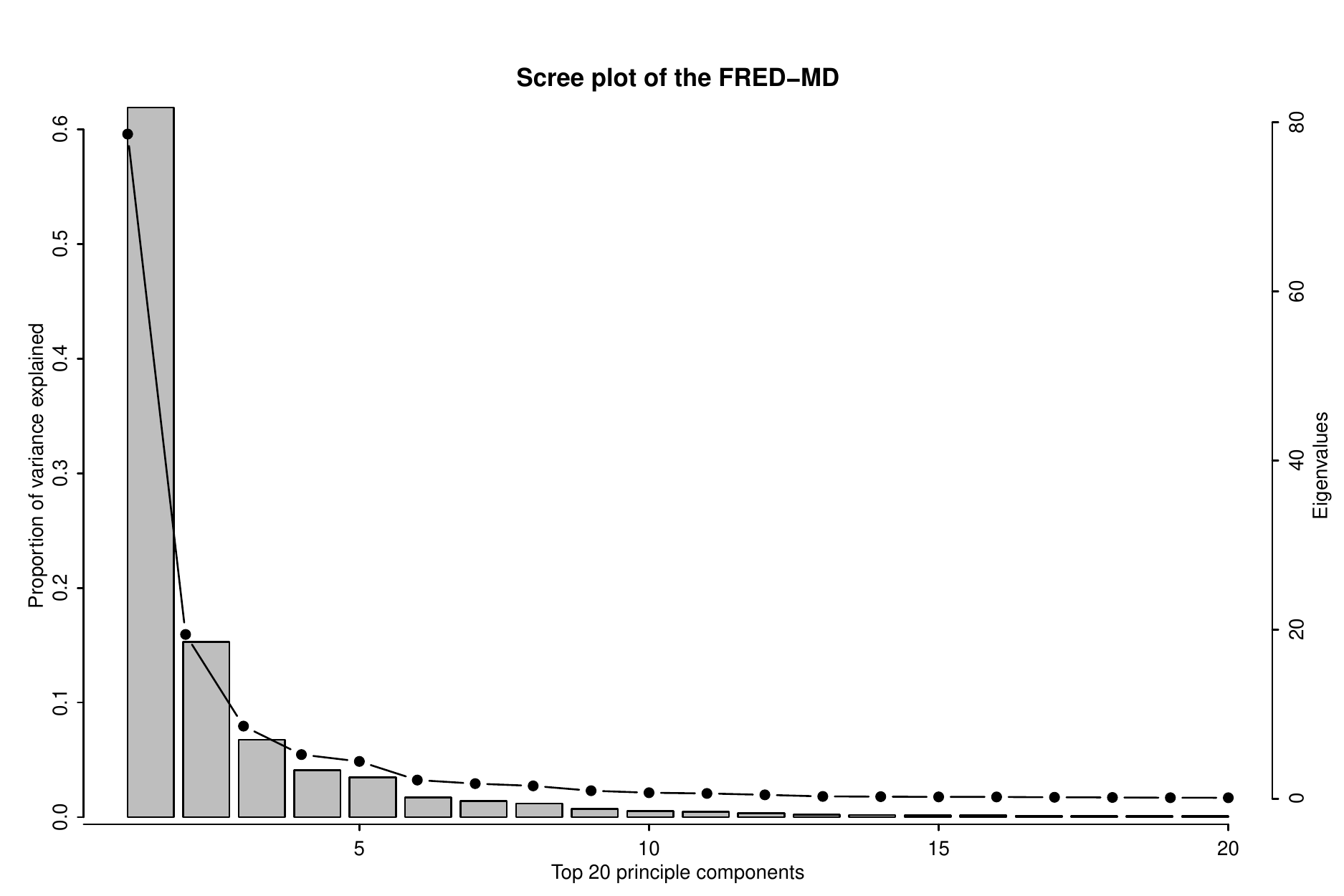}
  \caption{Eigenvalues (dotted line) and proportion of variance explained (bar) by the top 20 principal components}
 \label{Fig_scree_1}
\end{figure}

\begin{table}[htbp]
\begin{center}
\caption{Out of sample $R^2$ and average selected model size}
\label{tab_real_1}
\begin{tabular}{c|ccc|cc}
\hline\hline
Maturity of Bond & \multicolumn{3}{c|}{Out of sample $R^2$}
& \multicolumn{2}{c}{Average model size}
\\\hline
&  { FarmSelect} &  { LASSO}  &  { PCR} & { FarmSelect}
&  { Lasso}
\\
2 Years    &  0.2586   &  0.2295   &  0.2012   &  8.80    &  22.72
\\
3 Years    &  0.2295   &  0.2166   &  0.1854   &  8.92    &  21.40
\\
4 Years    &  0.2137   &  0.1801   &  0.1639   &  9.03     & 20.74
\\
5 Years    &  0.2004   &  0.1723   &  0.1463   &  9.21     &  20.21
\\\hline\hline
\end{tabular}
\end{center}
\end{table}

\begin{table}[htbp]
\begin{center}
\caption{2 years Maturity: Top 10 variables with highest selection frequency}
\label{tab_real_2}
\begin{tabular}{c|lc}
\hline\hline
& \multicolumn{2}{c|}{FarmSelect}
\\\hline
Rank &  Name & Frequency
\\
1    &  Switzerland / U.S. Foreign Exchange Rate & 75	
\\
2    & Civilians Unemployed - Less Than 5 Weeks
 & 73
\\
3    &  Moody's Baa Corporate Bond Minus FEDFUNDS
	& 72
\\
4    &  Housing Starts, Northeast
	& 71	
\\
5    &  Industrial Production: Durable Consumer Goods
	& 69
\\
6    &  CBOE S\&P 100 Volatility Index
    & 65	
\\
7    &  Real M2 Money Stock
	& 64	
\\
8    &  10-Year Treasury Rate
	& 63	
\\
9    &  CPI : Commodities
	& 61
\\
10    &  Commercial and Industrial Loans
	& 58 	
\\\hline
& \multicolumn{2}{c}{LASSO}
\\\hline
Rank &  Name & Frequency
\\
1     & CBOE S\&P 100 Volatility Index
	& 134
\\
2   	 & Industrial Production: Residential Utilities
	 & 130
\\
3    & Housing Starts, Northeast
	&127
\\
4   	& Switzerland / U.S. Foreign Exchange Rate
	&126
\\
5   	&Industrial Production: Fuels
	& 124
\\
6    & New Private Housing Permits, South
 & 122
\\
7   	& Canada / U.S. Foreign Exchange Rate
	& 117
\\
8    &10-Year Treasury Rate
	& 114
\\
9    &Japan / U.S. Foreign Exchange Rate
	&110
\\
10   & CPI : Commodities
	&106
\\\hline\hline
\end{tabular}
\end{center}
\end{table}

\begin{figure}[htbp]
 \centering
 \includegraphics[width=4.4 in]{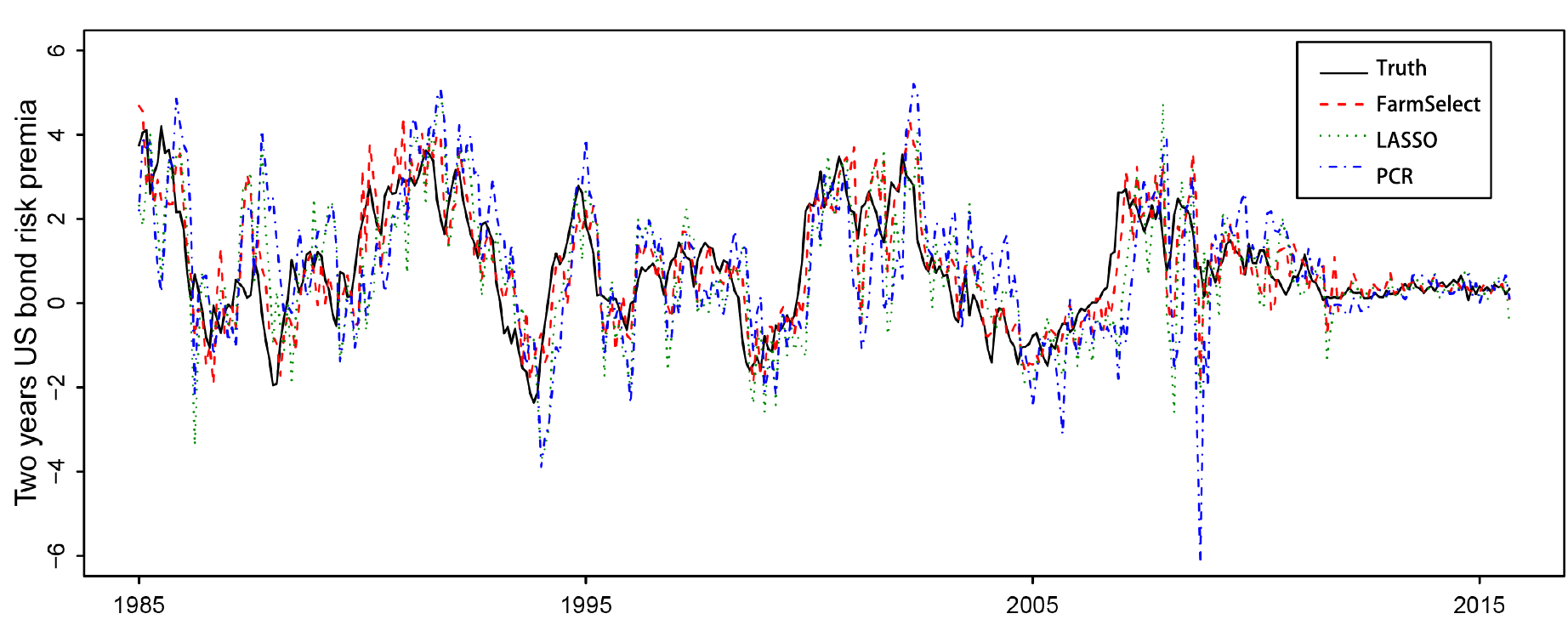}
 \includegraphics[width=4.4 in]{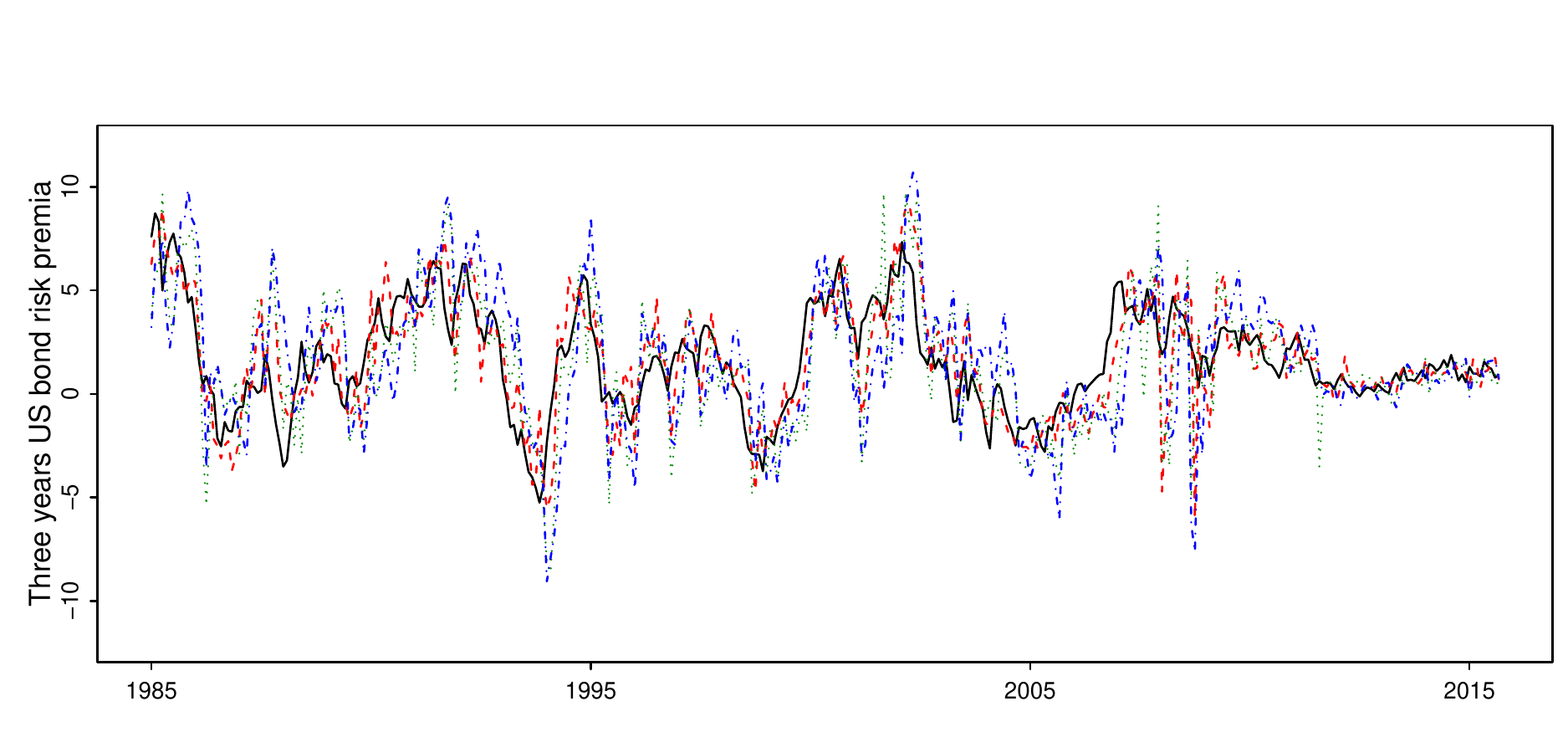}
 \includegraphics[width=4.4 in]{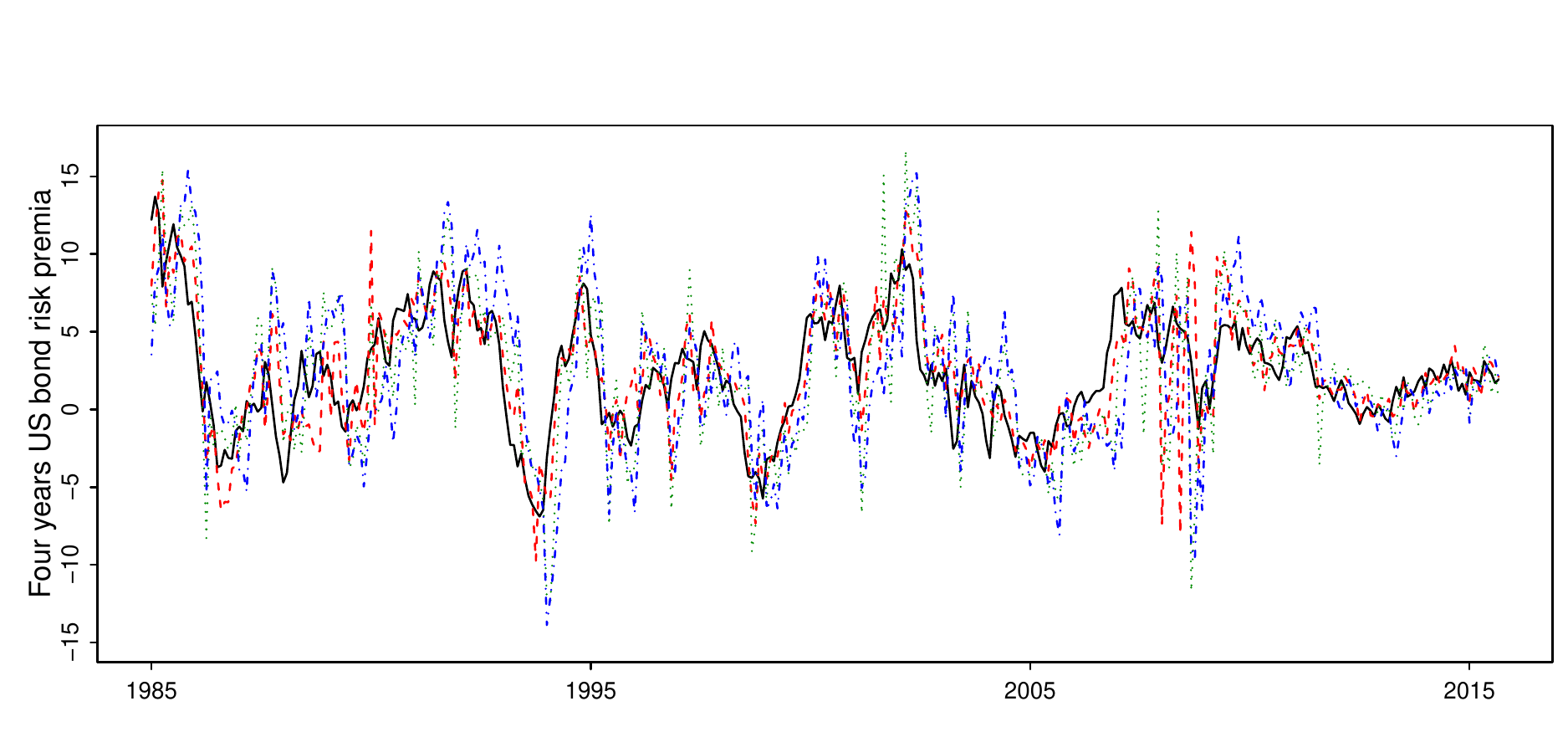}
 \includegraphics[width=4.4 in]{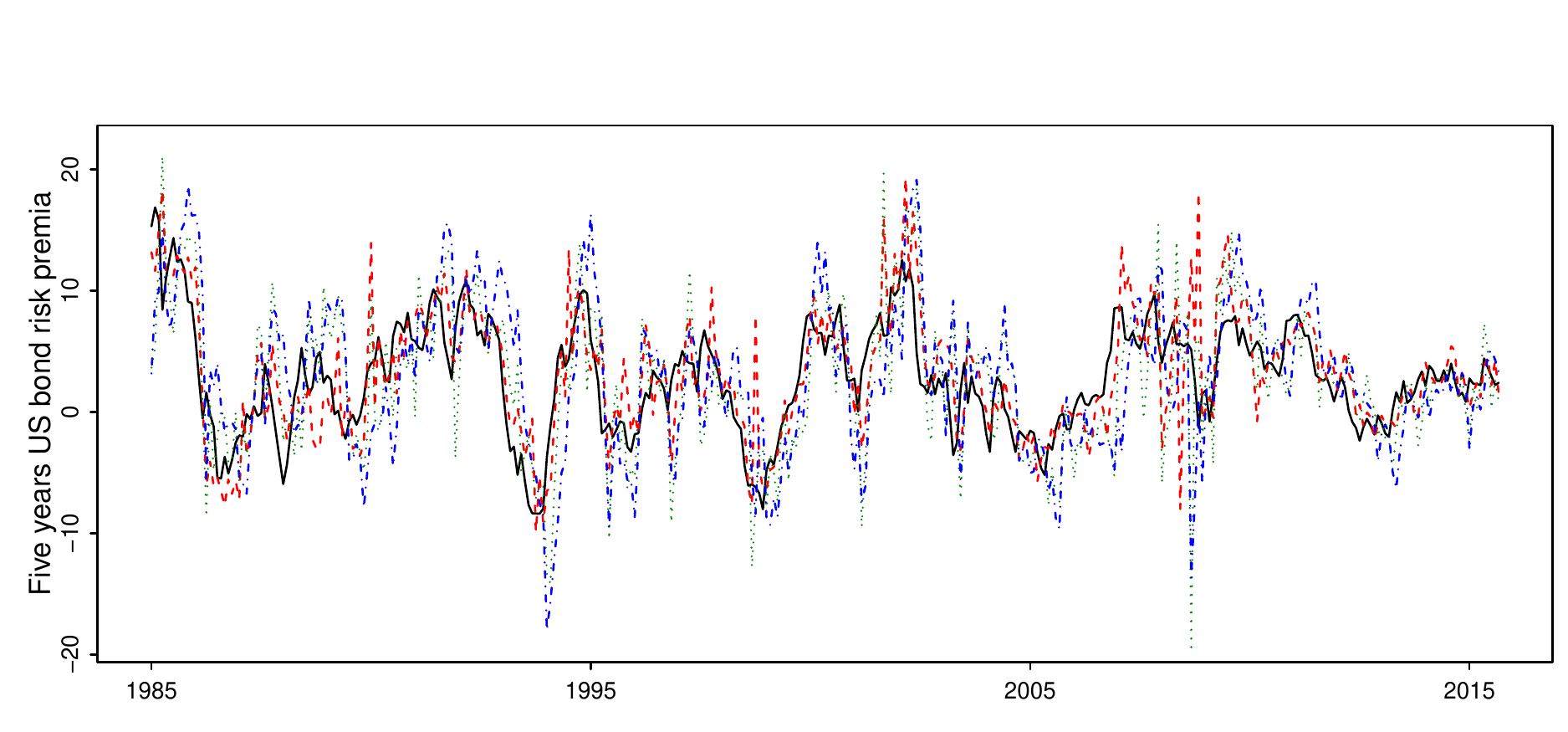}
   \caption{ One month ahead rolling window forecast. The window size is 120 months. The solid black line is the true bond risk premia. From top to bottom are the plots for 2 to 5 years bond risk premia}\label{fig_real_1}
\end{figure}

\newpage

\appendix
\section{Some preliminary results}\label{general-inverse-problems}
In the first appendix, we introduce some useful results in convex analysis and inverse problems. Under mild conditions, the tools we developed connect the unique global optimum of the regularized loss function $L_{\lambda}(\btheta)=L(\btheta)+\lambda R(\btheta)$ with the solution of a constrained problem $\min_{\rm{supp}(\btheta)\subseteq S}L_{\lambda}(\btheta)$.

\begin{lem}\label{constrainedproblem}
	Suppose $L(\btheta)\in C^2(\R^p)$ and is convex. $R(\btheta)$ is convex and $R(\balpha+\bbeta)=R(\balpha)+R(\bbeta)$ for $\balpha\in\M$ and $\bbeta\in\M^{\perp}$, where $\M$ is a linear subspace of $\R^p$ and $\M^{\perp}$ is its orthonormal complement. In addition, there exists $R^{\ast}(\btheta)\in C(\R^p)$ such that $|\langle \balpha,\bbeta\rangle|\leq R(\balpha)R^{\ast}(\bbeta)$ for $\balpha\in\M^{\perp}$ and $\bbeta\in\R^p$. Let $L_{\lambda}(\btheta)=L(\btheta)+\lambda R(\btheta)$ where $\lambda\geq 0$, and $\hat{\btheta}\in\argmin_{\btheta\in \M}L_{\lambda}(\btheta)$.
	
	If $R^{\ast}(\nabla L(\hat{\btheta}))<\lambda$ and $\btheta^T \nabla^2 L(\hat{\btheta})\btheta>0$ for all $\btheta\in\M$, then $\hat{\btheta}$ is the unique global minimizer of $L_{\lambda}(\btheta)$.
\end{lem}

\begin{proof}
	For any $\btheta\in\R^p$ we use $\btheta_{\M}$, $\btheta_{\M^{\perp}}$ to denote its orthonormal projections on $\M$ and $\M^{\perp}$, respectively.
	On the one hand, by convexity and orthogonality we have
	\begin{equation*}
	\begin{split}
	&L(\btheta)-L(\btheta_{\M})\geq \langle \nabla L(\btheta_{\M}),\btheta-\btheta_{\M}\rangle=\langle\nabla L(\btheta_{\M}),\btheta_{\M^{\perp}}\rangle\geq -R(\btheta_{
		\M^{\perp}})R^{\ast}(\nabla L(\btheta_{\M})).\\
	\end{split}
	\end{equation*}
	Since $R^{\ast}(\nabla L(\hat{\btheta}))<\lambda$, there exists $\delta>0$ such that $\|\btheta-\hat{\btheta}\|_2<\delta$ implies $R^{\ast}(\nabla L(\btheta))<\lambda$. Together with $\|\btheta_{\M}-\hat{\btheta}\|_2\leq\|\btheta-\hat{\btheta}\|_2$, we know $L(\btheta)-L(\btheta_{\M})\geq -\lambda R(\btheta_{\M^{\perp}})$ as long as $\|\btheta-\hat{\btheta}\|_2<\delta$, and the inequality strictly holds when $R(\btheta_{\M^{\perp}})>0$.
	
	On the other hand, $R(\btheta)-R(\btheta_{\M})=R(\btheta_{\M}+\btheta_{\M^{\perp}})-R(\btheta_{\M})=R(\btheta_{\M^{\perp}})$.
	Hence $\|\btheta-\hat{\btheta}\|_2<\delta$ forces $	L_{\lambda}(\btheta)-L_{\lambda}(\btheta_{\M})=[L(\btheta)-L(\btheta_{\M})]+\lambda[R(\btheta)-R(\btheta_{\M})]\geq 0$
	and the inequality strictly holds when $R(\btheta_{\M^{\perp}})>0$.
	
	Now suppose $0<\|\btheta-\hat{\btheta}\|_2<\delta$. If $\btheta\in\M$, then the facts $\hat{\btheta}\in\argmin_{\btheta\in \M}L_{\lambda}(\btheta)$ and $\balpha^T \nabla^2 L(\hat{\btheta})\balpha>0$, $\forall\balpha\in\M$ implies that $L_{\lambda}(\btheta')>L_{\lambda}(\hat{\btheta})$. In addition, our assumptions yield $\| \btheta\|_2^2\leq R(\btheta)R^{\ast}(\btheta)$ for $\btheta\in\M^{\perp}$, leading to $R(\btheta)>0$ over $\M^{\perp}\backslash\{0\}$. If $\btheta\notin\M$, then $R(\btheta_{\M^{\perp}})>0$ and $L_{\lambda}(\btheta)>L_{\lambda}(\btheta_{\M})\geq L_{\lambda}(\hat{\btheta})$. Therefore $\hat{\btheta}$ is a strict local optimum of $L_{\lambda}(\btheta)$, which is convex over $\R^p$. This finishes the proof.
	
\end{proof}

\medskip

\begin{lem}\label{boundarycontrol}
	Let $L(\btheta)$ be convex over a Euclidean space $\M$. If $\btheta_0\in\M$, $r>0$, and $L(\btheta)> L(\btheta_0)$ over the sphere $\partial B(\btheta_0,r)$, then any minimizer of $L(\btheta)$ is within the ball $B(\btheta_0,r)$.
\end{lem}

\begin{proof}
	For any $\btheta\notin \overline{B(\btheta_0,r)}$, there exists $t\in(0,1)$ and $\btheta'\in\partial B(\btheta_0,r)$ such that $\btheta'=(1-t)\btheta+t\btheta_0$. Then $L(\btheta_0)<L(\btheta')\leq (1-t)L(\btheta)+t L(\btheta_0)$, yielding $L(\btheta)>L(\btheta_0)$. Hence there is no minimizer outside $B(\btheta_0,r)$.
\end{proof}

\medskip

\begin{lem}\label{L2estimate-deterministic}
	Suppose $\M$ is a Euclidean space, $\btheta_0\in\M$ and $L(\btheta)$ is convex over $\M$. In addition, there exist $\kappa,A>0$ such that $L(\btheta)\geq L(\btheta_0)+ \langle \bh,\btheta-\btheta_0\rangle + \frac{\kappa}{2}\|\btheta-\btheta_0\|_2^2$
	as long as $\bh\in\partial L(\btheta_0)$ and $\|\btheta-\btheta_0\|_2\leq A$. If $\inf_{\bh\in\partial L(\btheta_0)}\|\bh\|_2<\frac{1}{2}\kappa A$, then any minimizer of $L_{\lambda}(\btheta)=L(\btheta)+\lambda R(\btheta)$ is within the ball $\{\btheta:\|\btheta-\btheta_0\|_2\leq\frac{2}{\kappa}\inf_{\bh\in\partial L(\btheta_0)}\|\bh\|_2\}$.
\end{lem}

\begin{proof}
	If $\|\btheta-\btheta_0\|_2<A$ and $\bh\in\partial L(\btheta_0)$, then
	\begin{equation*}
	\begin{split}
	&L(\btheta)-L(\btheta_0)\geq \langle \bh,\btheta-\btheta_0\rangle +\frac{\kappa}{2}\|\btheta-\btheta_0\|_2^2
	\geq -\|\bh\|_2\|\btheta-\btheta_0\|_2 +\frac{\kappa}{2}\|\btheta-\btheta_0\|_2^2\\
	&=\|\btheta-\btheta_0\|_2(\frac{\kappa}{2}\|\btheta-\btheta_0\|_2-\|\bh\|_2).
	\end{split}
	\end{equation*}
	Taking $\bh\in\partial L(\btheta_0)$ and $r>0$ such that $\frac{2}{\kappa}\|\bh\|_2<r<A$. This forces $L(\btheta)-L(\btheta_0)>0$ over the sphere $B(\btheta_0,r)$. Let $\hat{\btheta}$ be one of the minimizers of $L(\btheta)$. Lemma \ref{boundarycontrol} implies that $\|\hat{\btheta}-\btheta_0\|<r<A$. Then
	$0\geq L(\hat{\btheta})-L(\btheta_0)\geq\|\hat{\btheta}-\btheta_0\|_2(\frac{\kappa}{2}\|\hat{\btheta}-\btheta_0\|_2-\|\bh\|_2)$.
	The result is proved by taking the infimum over $\bh\in\partial L(\btheta_0)$.
\end{proof}

\medskip

\begin{cor}\label{L2estimate-deterministic-cor}
	Suppose $\lambda\geq 0$, $\M$ is a Euclidean space, $\btheta_0\in\M$, $L(\btheta)\in C^2(\M)$ and is convex, and $R(\btheta)$ is convex. In addition, there exist $\kappa,A>0$ such that $\nabla^2 L(\btheta)\succeq \kappa I$ as long as $\|\btheta-\btheta_0\|_2\leq A$. If $\|\nabla L(\btheta_0)\|_2+\lambda  \inf_{\bh\in\partial R(\btheta_0)}\|\bh\|_2<\frac{1}{2}\kappa A$, then $L_{\lambda}(\btheta)=L(\btheta)+\lambda R(\btheta)$ has unique minimizer $\hat{\btheta}$ and $\|\hat{\btheta}-\btheta_0\|_2\leq\frac{2}{\kappa}(\|\nabla L(\btheta_0)\|_2+\lambda \inf_{\bh\in\partial R(\btheta_0)}\|\bh\|_2)$.
\end{cor}

\begin{proof}
	Note that $\partial L_{\lambda}(\btheta_0)=\nabla L(\btheta_0)+\lambda\partial R(\btheta_0)$. There exists $\bh\in\partial R(\btheta_0)$ such that $\bh'=\nabla L(\btheta_0)+\bh\in \partial L_{\lambda}(\btheta_0)$ and $\|\bh'\|_2<\|\nabla L(\btheta_0)\|_2+\lambda\|\bh\|_2<\frac{1}{2}\kappa A$. Applying Lemma \ref{L2estimate-deterministic} to $L_{\lambda}$ and $\bh'$, we obtain that any minimizer of $L_{\lambda}$ satisfies $\|\hat{\btheta}-\btheta_0\|_2\leq\frac{2}{\kappa}\|\bh'\|_2<\frac{2}{\kappa}(\|\nabla L(\btheta_0)\|_2+\lambda\|\bh\|_2)$. Then $\|\hat{\btheta}-\btheta_0\|_2\leq A$ and $\nabla^2 L(\hat{\btheta})\succ 0$, proving both the bound and uniqueness.
\end{proof}

\subsection*{Proof of Lemma \ref{lem-identifiability} }
	Let $\bW = (\bw_1,\cdots,\bw_n)^T$ and
	$\btheta^*=(
	(\bbeta^*)^T, ( \bbeta^* )^T \bB_0 )^T$.
	Note that $ \nabla \E [ L_n( \by,\bW \btheta) ] = \E \{
	\frac{1}{n}\sum_{t=1}^{n}[-y_t+b( \bw_t^T \btheta )] \bw_t \}= \E \{
	[-y_1+b( \bw_1^T \btheta )] \bw_1 \}$ and $\bw_t^T \btheta^* = (1,\bx_t^T) \bbeta^*$. The claim is proved by
	\begin{equation*}
	\nabla \E [ L_n( \by,\bW \btheta) ] |_{\btheta=\btheta^*} = \E \{
	[-y_1+b( \bw_1^T \btheta^* )] \bw_1 \}
	=\E \{
	[-y_1+b( (1,\bx_t^T) \bbeta^* )] \bw_1 \}
	=\E (\eta_1\bw_1)=\mathbf{0}.
	\end{equation*}

\subsection*{Proof of Lemma \ref{lem-factor}}
Lemma \ref{lem-factor} is similar to the results developed in the Appendix C of \cite{WFa17} and hence  we omit the details.

\section{Proofs of Section \ref{theory}}
\subsection{Proof of Theorem \ref{consistency-general}}
Define $B_S(\btheta^*,r)=\{\btheta:\|\btheta-\btheta^*\|_2\leq r,\supp(\btheta)\subseteq S\}$ for $r>0$. We first introduce two useful lemmas.

\begin{lem}\label{inverse-perturbation}
	Suppose $\bA\in\R^{q\times r}$ and $\bB,\bC\in\R^{r\times r}$ and $\|\bC\bB^{-1}\|<1$, where $\|\cdot\|$ is an induced norm. Then $\|\bA[(\bB+\bC)^{-1}-\bB^{-1}]\|\leq \frac{\|\bA\bB^{-1}\|\cdot\|\bC\bB^{-1}\|}{1-\|\bC\bB^{-1}\|}$.
\end{lem}

\begin{proof}
	By the sub-multiplicity of induced norms,
	\begin{equation*}
	\begin{split}
	&\|\bA[(\bB+\bC)^{-1}-\bB^{-1}]\|=\|\bA\bB^{-1}[(\bI+\bC\bB^{-1})^{-1}-\bI]\|
	\leq \|\bA\bB^{-1}\|\cdot\|(\bI+\bC\bB^{-1})^{-1}-\bI\|\\
	&=\|\bA\bB^{-1}\|\cdot\Big\|\sum_{k=0}^{\infty}(-\bC\bB^{-1})^k-\bI \Big\|
	\leq \|\bA\bB^{-1}\|\sum_{k=1}^{\infty}\|\bC\bB^{-1}\|^k=\frac{\|\bA\bB^{-1}\|\cdot\|\bC\bB^{-1}\|}{1-\|\bC\bB^{-1}\|}.
	\end{split}
	\end{equation*}
\end{proof}

\medskip

\begin{lem}\label{lem-rsc}
	Under Assumptions \ref{assump-smoothness} and \ref{assump-RSC},
	we have $\|(\nabla^2_{SS}L_n(\btheta))^{-1}\|_2<\kappa_2^{-1}$ and $\|(\nabla^2_{SS}L_n(\btheta))^{-1}\|_{\infty}<\kappa_{\infty}^{-1}$ over $B_S(\btheta^*,\min\{A, \frac{\kappa_{\infty}}{M}\})$.
\end{lem}

\begin{proof}
	Define $\alpha_p(\btheta)=\|(\nabla^2_{SS}L_n(\btheta^*))^{-1}[\nabla^2_{SS}L_n(\btheta)-\nabla^2_{SS}L_n(\btheta^*)]\|_p$ for $p\in\{2,\infty\}$ and $\btheta\in B_S(\btheta^*,A)$.
	Note that for any symmetric matrix $\bA$, we have $\|\bA\|_1=\|\bA\|_{\infty}$ and $\|\bA\|_2\leq\sqrt{\|\bA\|_1\|\bA\|_{\infty}}\leq\|\bA\|_{\infty}$.
	Hence by the Assumptions we obtain that when $\|\btheta-\btheta^*\|_2\leq \min\{A, \frac{\kappa_{\infty}}{M}\}$ and $p\in\{2,\infty\}$,
	\begin{equation*}
	\begin{split}
	&\alpha_p(\btheta)\leq \|(\nabla^2_{SS}L_n(\btheta^*))^{-1}\|_{\infty}\|\nabla^2_{SS}L_n(\btheta)-\nabla^2_{SS}L_n(\btheta^*)\|_{\infty}
	\leq \frac{1}{2\kappa_{\infty}}M\|\btheta-\btheta^*\|_2\leq \frac{1}{2}.\\
	\end{split}
	\end{equation*}
	
	Lemma \ref{inverse-perturbation} leads to
	\begin{equation*}
	\begin{split}
	&\|(\nabla^2_{SS}L_n(\btheta))^{-1}-(\nabla^2_{SS}L_n(\btheta^*))^{-1}\|_{\infty}
	\leq \|(\nabla^2_{SS}L_n(\btheta^*))^{-1}\|_{\infty}\frac{\alpha_{\infty}}{1-\alpha_{\infty}}
	<\frac{1}{2\kappa_{\infty}},\\
	&\|(\nabla^2_{SS}L_n(\btheta))^{-1}-(\nabla^2_{SS}L_n(\btheta^*))^{-1}\|_2
	\leq \|(\nabla^2_{SS}L_n(\btheta^*))^{-1}\|_2\frac{\alpha_2}{1-\alpha_2}
	<\frac{1}{2\kappa_2}.\\
	\end{split}
	\end{equation*}
	Then the proof is finished by triangle's inequality and Assumption \ref{assump-RSC}.
\end{proof}

\medskip

Now we are ready to prove Theorem \ref{consistency-general}.
\begin{proof}[\textbf{\textit Proof of Theorem \ref{consistency-general}}]

First we study the restricted problem $\bar{\btheta}=\argmin_{\btheta\in\M}\{L_n(\btheta)+\lambda R(\btheta)\}$.
Take $R(\btheta)=\|\btheta_{S_{[p]}}\|_1$ and $R^{\ast}(\btheta)=\|\btheta_{S_2}\|_{\infty}$.
Let $A_1=\min\{A,\frac{\kappa_{\infty}\tau}{3M}\}$ and hence $A_1\leq\min\{A,\frac{\kappa_{\infty}}{M}\}$. Lemma \ref{lem-rsc} shows that $\|(\nabla^2_{SS}L_n(\btheta))^{-1}\|_2<\kappa_2^{-1}$ and $\|(\nabla^2_{SS}L_n(\btheta))^{-1}\|_{\infty}<\kappa_{\infty}^{-1}$ over $B_S(\btheta^*,A_1)$.
	
	Since $\supp(\btheta^*)\subseteq S$, any $\bh\in\partial R(\btheta^*)$ satisfies $\|\bh\|_2\leq\sqrt{|S_1|}$. Therefore
$$\|\nabla_S L_n(\btheta^*)\|_2+\lambda\|\bh\|_2\leq \frac{1}{2}\kappa_2 A_1\leq\frac{1}{2}\kappa_2 A.$$  Then Corollary \ref{L2estimate-deterministic-cor} implies that $\|\bar{\btheta}-\btheta^{\ast}\|_2\leq\frac{2}{\kappa_2}(\|\nabla_S L(\btheta^*)\|_2+\lambda\sqrt{|S_1|})\leq A_1$.
	
\medskip	
	
	 Second, we study the $L^{\infty}$ bound. On the one hand, the optimality condition yields $\nabla_SL_n(\bar{\btheta})\in\lambda \partial\|\bar{\btheta}_{[p]}\|_{\infty}$ and hence $\|\nabla_SL_n(\bar{\btheta})\|_{\infty}\leq \lambda$. On the other hand, by letting $\btheta_t=(1-t)\btheta^*+t\bar{\btheta}$ ($0\leq t\leq 1$) we have
	\begin{equation*}
	\begin{split}
	&\nabla_S L_n(\bar{\btheta})-\nabla_S L_n(\btheta^*)=\int_{0}^{1}\nabla_{SS}^2 L_n(\btheta_t)(\bar{\btheta}-\btheta^*)\mathrm{d} t\\
	&=\nabla_{SS}^2 L_n(\btheta^*)(\bar{\btheta}-\btheta^*)
	+\int_{0}^{1}[\nabla_{SS}^2 L_n(\bar{\btheta}_t)-\nabla_{SS}^2 L_n(\btheta^*)](\bar{\btheta}-\btheta^*)\mathrm{d} t.\\
	\end{split}
	\end{equation*}
	Hence
	\begin{equation*}
	\begin{split}
	&\|(\bar{\btheta}-\btheta^*)-(\nabla_{SS}^2 L_n(\btheta^*))^{-1}
	[\nabla_S L_n(\bar{\btheta})-\nabla_S L_n(\btheta^*)]\|_{\infty}\\
	&\leq
	\int_{0}^{1}\|(\nabla_{SS}^2 L_n(\btheta^*))^{-1}[\nabla_{SS}^2 L_n(\bar{\btheta}_t)-\nabla_{SS}^2 L_n(\btheta^*)](\bar{\btheta}-\btheta^*)\|_{\infty}\mathrm{d} t\\
	&\leq
	\|(\nabla_{SS}^2 L_n(\btheta^*))^{-1}\|_{\infty}\sup_{t\in[0,1]}\|\nabla_{SS}^2 L_n(\bar{\btheta}_t)-\nabla_{SS}^2 L_n(\btheta^*)\|_{\infty}\|\bar{\btheta}-\btheta^*\|_{\infty}\\
	\end{split}
	\end{equation*}
	By Assumptions \ref{assump-smoothness} and \ref{assump-RSC}, we obtain that
	\begin{equation*}
	\begin{split}
	&\|(\bar{\btheta}-\btheta^*)-(\nabla_{SS}^2 L_n(\btheta^*))^{-1}
	[\nabla_S L_n(\bar{\btheta})-\nabla_S L_n(\btheta^*)]\|_{\infty}\leq
	\frac{M}{2\kappa_{\infty}}\|\bar{\btheta}-\btheta^*\|_2\|\bar{\btheta}-\btheta^*\|_{\infty}.\\
	\end{split}
	\end{equation*}
	By $\bar{\btheta}\in B_S(\btheta^*,A_1)$ we have
	\begin{equation*}
	\begin{split}
	&\|\bar{\btheta}-\btheta^*\|_{\infty}\leq\|(\nabla_{SS}^2 L_n(\btheta^*))^{-1}\|_{\infty}
	\|\nabla_S L_n(\bar{\btheta})-\nabla_S L_n(\btheta^*)\|_{\infty}+
	\frac{M}{2\kappa_{\infty}}\|\bar{\btheta}-\btheta^*\|_2\|\bar{\btheta}-\btheta^*\|_{\infty}\\
	&\leq \frac{1}{2\kappa_{\infty}}(\lambda+\|\nabla_S L_n(\btheta^*)\|_{\infty})+\frac{1}{6}\|\bar{\btheta}-\btheta^*\|_{\infty}.
	\end{split}
	\end{equation*}
	Therefore,
	\begin{equation}\label{Thm_4_1_tnfty}
	|\bar{\btheta}-\btheta^{\ast}\|_{\infty}\leq \frac{3}{5\kappa_{\infty}}(\|\nabla_S L_n(\btheta^*)\|_{\infty}+\lambda).
	\end{equation}
	
\medskip	
	
	Third we study the $L^1$ bound. The bound on $\|\bar{\btheta}-\btheta^*\|_1$ can be obtained in a similar way. Using the fact that $\|\cdot\|_1=\|\cdot\|_{\infty}$  for symmetric matrices,
	\begin{equation*}
	\begin{split}
	&\|\bar{\btheta}-\btheta^*\|_1\leq\|(\nabla_{SS}^2 L_n(\btheta^*))^{-1}\|_1
	\|\nabla_S L_n(\bar{\btheta})-\nabla_S L_n(\btheta^*)\|_1+
	\frac{M}{2\kappa_{\infty}}\|\bar{\btheta}-\btheta^*\|_2\|\bar{\btheta}-\btheta^*\|_1\\
	&\leq \frac{1}{2\kappa_{\infty}}(\lambda|S_1|+\|\nabla_S L_n(\btheta^*)\|_1)+\frac{1}{6}\|\bar{\btheta}-\btheta^*\|_1.
	\end{split}
	\end{equation*}
	Hence $\|\bar{\btheta}-\btheta^{\ast}\|_1\leq \frac{3}{5\kappa_{\infty}}(\|\nabla_S L_n(\btheta^*)\|_1+\lambda|S_1|)$.  Since $\supp(\bar{\btheta})\subseteq S$, we also have
	$$\|\bar{\btheta}-\btheta^{\ast}\|_1\leq\sqrt{|S|}\|\bar{\btheta}-\btheta^{\ast}\|_2\leq \frac{2\sqrt{|S|}}{\kappa_2}(\|\nabla_S L(\btheta^*)\|_2+\lambda\sqrt{|S_1|}).$$
	This gives another $L^1$ bound.
	
\medskip	
	
	By Lemma \ref{constrainedproblem}, to derive $\hat{\btheta}=\bar{\btheta}$ it remains to show that $\|\nabla_{S_2}L_n(\bar{\btheta})\|_{\infty}<\lambda$. Using the Taylor expansion we have
	\begin{equation}\label{Thm_4_1_2}
	\begin{split}
	& \nabla_{S_2}L_n(\bar{\btheta})-\nabla_{S_2}L_n(\btheta^*)=\int_{0}^{1}\nabla_{S_2 S}^2L_n(\btheta_t)(\bar{\btheta}-\btheta^*) {\rm d}t\\
	&=\nabla_{S_2 S}^2L_n(\btheta^*)(\bar{\btheta}-\btheta^*)
	+\int_{0}^{1}[\nabla^2_{S_2 S}L_n(\btheta_t)-\nabla^2_{S_2 S}L_n(\btheta^*)](\bar{\btheta}-\btheta^*){\rm d}t.
	\end{split}
	\end{equation}
	On the one hand, the first term in (\ref{Thm_4_1_2}) follows,
	\begin{equation*}
	\begin{split}
	&\|\nabla_{S_2 S}^2L_n(\btheta^*)(\bar{\btheta}-\btheta^*)\|_{\infty}
	=\|[\nabla_{S_2 S}^2L_n(\btheta^*)(\nabla_{SS}^2 L_n(\btheta^*))^{-1}][\nabla_{SS}^2 L_n(\btheta^*)(\bar{\btheta}-\btheta^*)]\|_{\infty}\\
	&\leq (1-\tau)\|\nabla_{SS}^2 L_n(\btheta^*)(\bar{\btheta}-\btheta^*)\|_{\infty}.
	\end{split}
	\end{equation*}
	By the Taylor expansion, triangle's inequality, Assumption \ref{assump-smoothness} and the fact that $\bar{\btheta}\in B_S(\btheta^*,A_1)$,
	\begin{equation*}
	\begin{split}
	&\|\nabla_{SS}^2 L_n(\btheta^*)(\bar{\btheta}-\btheta^*)\|_{\infty}\leq\|\nabla_S L_n(\bar{\btheta})-\nabla_S L_n(\btheta^*)\|_{\infty}+
	\int_{0}^{1}\|[\nabla_{SS}^2 L_n(\bar{\btheta}_t)-\nabla_{SS}^2 L_n(\btheta^*)](\bar{\btheta}-\btheta^*)\|_{\infty}\mathrm{d} t\\
	&\leq \|\nabla_S L_n(\bar{\btheta})\|_{\infty}+\|\nabla_S L_n(\btheta^*)\|_{\infty}+
	M\|\bar{\btheta}-\btheta^*\|_2\|\bar{\btheta}-\btheta^*\|_{\infty}\\
	&\leq \lambda+\|\nabla_S L_n(\btheta^*)\|_{\infty}+\frac{\kappa_{\infty}\tau}{3}\|\bar{\btheta}-\btheta^*\|_{\infty}.
	\end{split}
	\end{equation*}
	On the other hand, we bound the second term in (\ref{Thm_4_1_2}). Note that $\btheta_t\in B_S(\btheta^*,A_1)$ for all $t\in[0,1]$. Assumption \ref{assump-smoothness} yields
	\begin{equation*}\label{proof-ic-0}
	\begin{split}
	&\Big\|\int_{0}^{1}[\nabla^2_{S_2 S}L_n(\btheta_t)-\nabla^2_{S_2 S}L_n(\btheta^*)](\bar{\btheta}-\btheta^*){\rm d}t\Big\|_{\infty}\\
	&\leq
	\sup_{t\in[0,1]}\|\nabla^2_{S_2 S}L_n(\btheta_t)-\nabla^2_{S_2 S}L_n(\btheta^*)\|_{\infty}\|\bar{\btheta}-\btheta^*\|_{\infty}
	\leq \frac{\kappa_{\infty}\tau}{3}\|\bar{\btheta}-\btheta^*\|_{\infty}.
	\end{split}
	\end{equation*}
	As a result,
	\begin{equation*}
	\begin{split}
	& \|\nabla_{S_2}L_n(\bar{\btheta})\|_{\infty}
	\leq\|\nabla_{S_2}L_n(\btheta^*)\|_{\infty}+(1-\tau)\Big(\lambda+\|\nabla_S L_n(\btheta^*)\|_{\infty}+\frac{\kappa_{\infty}\tau}{3}\|\bar{\btheta}-\btheta^*\|_{\infty}\Big)+\frac{\kappa_{\infty}\tau}{3}\|\bar{\btheta}-\btheta^*\|_{\infty}\\
	&\leq \lambda-\tau\Big(\lambda-\frac{2\kappa_{\infty}}{3}\|\bar{\btheta}-\btheta^*\|_{\infty}-\frac{2}{\tau}\|\nabla L_n(\btheta^*)\|_{\infty}\Big).
	\end{split}
	\end{equation*}
	Recall that the $L^{\infty}$ bound in (\ref{Thm_4_1_tnfty}). By plugging in this estimate, and using the assumptions $0<\tau<1$ and $\lambda>\frac{20}{3\tau}\|\nabla L_n(\btheta^*)\|_{\infty}$, we derive that
	\begin{equation*}
	\begin{split}
	& \|\nabla_{S_2}L_n(\bar{\btheta})\|_{\infty}
	\leq \lambda-\tau\Big(\lambda-\frac{2}{5}(\|\nabla_S L_n(\btheta^*)\|_{\infty}+\lambda)-\frac{2}{\tau}\|\nabla L_n(\btheta^*)\|_{\infty}\Big)\\
	&\leq \lambda-\tau\Big(\frac{3}{5}\lambda-\frac{4}{\tau}\|\nabla L_n(\btheta^*)\|_{\infty}
	\Big)<\lambda.
	\end{split}
	\end{equation*}
	This implies $\hat{\btheta}=\bar{\btheta}$ and translates all the bounds for $\bar{\btheta}$ to the ones for $\hat{\btheta}$. The proposition on sign consistency follows from elementary computation, thus we omit its proof.
\end{proof}

\subsection{Proof of Theorem \ref{consistency-estimated-factors}}
\begin{proof}[\textbf{\textit Proof of Theorem \ref{consistency-estimated-factors}}]
	Recall that $\hat{\btheta}=\argmin_{\btheta}\{L_n(\by,\hat{\bW}\btheta)+\lambda \|\btheta_{[p]}\|_1\}$. Also, Assumption \ref{assump-factor} tells us $\bH_0$ is nonsingular and so is $\bH=\begin{pmatrix}
\bI_p & \mathbf{0}_{p\times K}\\
	\mathbf{0}_{K\times p} & 	\bH_0
	\end{pmatrix}$. Define $\overline{\bW}=\hat{\bW}\bH$, $\bar{\btheta}=\bH^{-1}\hat{\btheta}$, $\hat{\bB}_0=(\mathbf{0}_K^T,\hat\bB^T)^T$,
	$\hat{\btheta}^*=\begin{pmatrix}
 \bbeta^*\\
\hat{\bB}_0^T \bbeta^*
	\end{pmatrix}$
	 and $\bar{\btheta}^*=
	\bH^{-1}\hat\btheta^*$. We easily see that $\hat{\bbeta}=\hat{\btheta}_{[p]}=\bar{\btheta}_{[p]}$ and $\bar{\btheta}=\argmin_{\btheta}\{L_n(\by,\overline{\bW}\btheta)+\lambda \|\btheta_{[p]}\|_1\}$. Then it follows that $\supp(\hat{\bbeta})=\supp(\bar{\btheta}_{[p]})$ and $\|\hat{\bbeta}-\bbeta^*\|=\|\bar{\btheta}_{[p]}-\bar{\btheta}^*_{[p]}\|\leq \|\bar{\btheta}-\bar{\btheta}^*\|$ for any norm $\|\cdot\|$.
	
	Consequently, Theorem \ref{consistency-estimated-factors} is reduced to studying $\bar{\btheta}$ and the loss function $L_n(\by,\overline{\bW}\btheta)$. The Lemma \ref{lem-bridge} below implies that all the regularity conditions (with $A=\infty$) in Theorem \ref{consistency-general} are satisfied.
	
	
	Let $w_{tj}$ and $\overline{w}_{tj}$ be the $(i,j)$-th element of $\bW$ and $\overline{\bW}$, respectively.
	Observe that $L_n(\by,\overline{\bW}\btheta)=\frac{1}{n}\sum_{t=1}^{n}[-y_t\overline{\bw}_t^T\btheta+b(\overline{\bw}_t^T\btheta)]$, $\nabla L_n(\by,\overline{\bW}\btheta)=\frac{1}{n}\sum_{t=1}^{n}[-y_t+b'(\overline{\bw}_t^T\btheta)]\overline{\bw}_t$ and $\overline{\bW}\bar{\btheta}^*=\bX_1 \bbeta^*$.
	Hence $\|\nabla L_n(\by,\overline{\bW}\bar{\btheta}^*)\|_{\infty}= \varepsilon$ and consequently, $\|\nabla_S L_n(\by,\overline{\bW}\bar{\btheta}^*)\|_{\infty}\leq\varepsilon$, $\|\nabla_S L_n(\by,\overline{\bW}\bar{\btheta}^*)\|_2\leq\varepsilon\sqrt{|S|}$ and $\|\nabla_S L_n(\by,\overline{\bW}\bar{\btheta}^*)\|_1\leq\varepsilon|S|$. In addition, $\lambda>7\varepsilon/\tau\geq\varepsilon$. Based on these estimates, all the results follow from Theorem \ref{consistency-general} and some simple algebra.
\end{proof}

Here we present the Lemma \ref{lem-bridge} used above and its proof.

\begin{lem}\label{lem-bridge}
	Let Assumptions \ref{assump-regularity-domain}, \ref{assump-RSC-1} and \ref{assump-factor} hold. Treat $L_n(\by,\overline{\bW}\btheta)$ as a function of $\btheta$, and the derivatives below are taken with respect to it. Define $M=M_0^3M_3|S|^{3/2}$. Then
	\begin{align*}
	(i) \quad &\|\nabla_{\cdot S}^2 L(\by,\overline{\bW}\btheta)-\nabla_{\cdot S}^2 L(\by,\overline{\bW}\bar{\btheta}^*)\|_{\infty}
	\leq M\|\btheta-\bar{\btheta}^*\|_2,~\forall \btheta,\\
	(ii) \quad &\|(\nabla^2_{SS}L(\by,\overline{\bW}\bar{\btheta}^*))^{-1}\|_{\infty}\leq \frac{1}{2\kappa_{\infty}},\\
	(iii) \quad &\|(\nabla^2_{SS}L(\by,\overline{\bW}\bar{\btheta}^*))^{-1}\|_2\leq \frac{1}{2\kappa_2},\\
	(iv) \quad & \|\nabla^2_{S_2 S} L(\by,\overline{\bW}\bar{\btheta}^{\ast})(\nabla_{S S}^2 L(\by,\overline{\bW}\bar{\btheta}^{\ast}))^{-1}\|_{\infty}\leq 1-\tau.
	\end{align*}
	
\end{lem}

\begin{proof}
	(i) Based on the fact that $\bW \btheta^*=\overline{\bW} \bar{\btheta}^*=\bX_1 \bbeta^*$, we have $\nabla^2L_n(\by,\bW\btheta^*)
	=\frac{1}{n}\sum_{t=1}^{n}b''( \overline{\bw}_t^T\bar{\btheta}^*)\bw_{t}\bw_{t}^T$ and $\nabla^2L_n(\by,\overline{\bW}\bar{\btheta}^*)
	=\frac{1}{n}\sum_{t=1}^{n}b''(\overline{\bw}_t^T\bar{\btheta}^*)\overline{\bw}_{t}\overline{\bw}_{t}^T$.
	For any $j,k\in[p+K]$ and $\supp(\btheta)\subseteq S$,
	\begin{equation}\label{Lemma_B_3_01}
	\begin{split}
	&|\nabla^2_{jk}L_n(\by,\overline{\bW}\btheta)-\nabla^2_{jk}L_n(\by,\overline{\bW}\bar{\btheta}^*)|
	\leq\frac{1}{n}\sum_{t=1}^{n}|b''(\overline{\bw}_t^T\btheta)-b''(\overline{\bw}_t^T\bar{\btheta}^*)|\cdot|\overline{\bw}_{tj}\overline{\bw}_{tk}|\\
	&\leq \frac{1}{n}\sum_{t=1}^{n} M_3 |\overline{\bw}_t^T(\btheta-\bar{\btheta}^*)|\cdot\|\overline{\bW}\|_{\max}^2\\
	\end{split}
	\end{equation}
	By the Cauchy-Schwarz inequality and $\|\overline{\bW}\|_{\max}\leq \|\bW\|_{\max}+\|\overline{\bW}-\bW\|_{\max}\leq M_0$, we obtain that for $i\in[n]$,
$ |\overline{\bw}_t^T(\btheta-\bar{\btheta}^*)|=|\overline{\bw}_{tS}^T(\btheta-\bar{\btheta}^*)_S|
	\leq \|\overline{\bw}_{tS}\|_2\|\btheta-\bar{\btheta}^*\|_2
	\leq \sqrt{|S|} M_0 \|\btheta-\bar{\btheta}^*\|_2$.
	Plugging this result back to (\ref{Lemma_B_3_01}), we get
	\begin{equation*}
	\begin{split}
	&|\nabla^2_{jk}L_n(\by,\overline{\bW}\btheta)-\nabla^2_{jk}L_n(\by,\overline{\bW}\bar{\btheta}^*)|
	\leq \sqrt{|S|}M_3 M_0^3\|\btheta-\bar{\btheta}^*\|_2,~~\forall j,k\in[p+K];\\
	&\|\nabla^2_{\cdot S}L_n(\by,\overline{\bW}\btheta)-\nabla^2_{\cdot S}L_n(\by,\overline{\bW}\bar{\btheta}^*)\|_{\infty}
	\leq |S|^{3/2}M_3 M_0^3\|\btheta-\bar{\btheta}^*\|_2=M\|\btheta-\bar{\btheta}^*\|_2.\\
	\end{split}
	\end{equation*}
	
	\medskip
	
	(ii) 	Now we come to the second claim. For any $k\in[p+K]$,
	\begin{equation*}
	\begin{split}
	&\|\nabla^2_{kS}L_n(\by,\overline{\bW}\bar{\btheta}^*)-\nabla^2_{kS}L_n(\by,\bW\btheta^*)\|_{\infty}
	\leq\frac{1}{n}\sum_{t=1}^{n}b''(\bx_t^T\bbeta^*)\|\overline{w}_{tk}\overline{\bw}_{tS}^T-w_{tk}\bw_{tS}^T\|_{\infty}\\
	&\leq \frac{M_2\sqrt{|S|}}{n}\sum_{t=1}^{n}\|\overline{w}_{tk}\overline{\bw}_{tS}^T-w_{tk}\bw_{tS}^T\|_{2}.
	\end{split}
	\end{equation*}
	Also, by $\|\bW\|_{\max}\leq M_0/2$ and $\|\overline{\bW}\|_{\max}\leq M_0$ we have
	\begin{equation*}
	\begin{split}
	&\|\overline{w}_{tk}\overline{\bw}_{tS}^T-w_{tk}\bw_{tS}^T\|_2
	\leq | w_{tk}|\cdot\|(\overline{\bw}_{tS}-\bw_{tS})^T\|_2+|\overline{w}_{tk}-w_{tk}|\cdot\|\overline{\bw}_{tS}^T\|_2\\
	&\leq \|\bW\|_{\max}\|\overline{\bw}_{tS}-\bw_{tS}\|_2+|\overline{w}_{tk}-w_{tk}|\cdot \sqrt{|S|}\|\overline{\bW}\|_{\max}\\
	&\leq \frac{M_0}{2}\|\overline{\bw}_{tS}-\bw_{tS}\|_2+M_0\sqrt{|S|}\cdot|\overline{w}_{tk}-w_{tk}|.
	\end{split}
	\end{equation*}
	Define $\delta=\max_{j\in[p+K]}(\frac{1}{n}\sum_{t=1}^{n}|\overline{w}_{tj}-w_{tj}|^2)^{1/2}$. By the Jensen's inequality, $\forall J\subseteq[p+K]$,
	\begin{equation*}
	\begin{split}
	&\frac{1}{n}\sum_{t=1}^{n}\|\overline{\bw}_{tJ}-\bw_{tJ}\|_2\leq \Big(\frac{1}{n}\sum_{t=1}^{n}\|\overline{\bw}_{tJ}-\bw_{tJ}\|_2^2\Big)^{1/2}
	\leq \Big(\frac{|J|}{n}\max_{j\in[p+K]}\sum_{t=1}^{n}|\overline{w}_{tj}-w_{tj}|^2\Big)^{1/2}\leq\sqrt{|J|}\delta.
	\end{split}
	\end{equation*}
	
	As a result,
	\begin{equation}\label{lem-B-inf}
	\begin{split}
	&\|\nabla^2_{\cdot S}L_n(\by,\overline{\bW}\bar{\btheta}^*)-\nabla^2_{\cdot S}L_n(\by,\bW\btheta^*)\|_{\infty}\\
	&=\max_{k\in[p+K]}
	\|\nabla^2_{kS}L_n(\by,\overline{\bW}\bar{\btheta}^*)-\nabla^2_{kS}L_n(\by,\bW\btheta^*)\|_{\infty}\leq \frac{3}{2}M_0M_2|S|\delta.
	\end{split}
	\end{equation}
	
	Let $\alpha=\|(\nabla^2_{SS}L_n(\by,\bW\btheta^*))^{-1}[\nabla^2_{SS}L_n(\by,\overline{\bW}\bar{\btheta}^*)-\nabla^2_{SS}L_n(\by,\bW\btheta^*)]\|_{\infty}$. Then
	\begin{equation}\label{lemma_B_3_alpha}
	\begin{split}
	&\alpha\leq \|(\nabla^2_{SS}L_n(\by,\bW\btheta^*))^{-1}\|_{\infty}\|\nabla^2_{SS}L_n(\by,\overline{\bW}^T\bar{\btheta}^*)-\nabla^2_{SS}L_n(\by,\bW\btheta^*)\|_{\infty}\\
	&\leq \frac{3}{8\kappa_{\infty}}M_0M_2|S|\delta\leq\frac{1}{2}.\\
	\end{split}
	\end{equation}
	
	Lemma \ref{inverse-perturbation} yields
	\begin{equation*}
	\begin{split}
	&\|(\nabla^2_{SS}L_n(\by,\overline{\bW}\bar{\btheta}^*))^{-1}-(\nabla^2_{SS}L_n(\by,\bW\btheta^*))^{-1}\|_{\infty}
	\leq \|(\nabla^2_{SS}L_n(\by,\bW\btheta^*))^{-1}\|_{\infty}\frac{\alpha}{1-\alpha}\\
	&\leq\frac{1}{4\kappa_{\infty}}\cdot\frac{\alpha}{1-\frac{1}{2}} \leq \frac{3}{16\kappa_{\infty}^2}M_0M_2|S|\delta.
	\end{split}
	\end{equation*}
	We also have a cruder bound $\|(\nabla^2_{SS}L_n(\by,\overline{\bW}\bar{\btheta}^*))^{-1}-(\nabla^2_{SS}L_n(\by,\bW\btheta^*))^{-1}\|_{\infty}\leq \frac{1}{4\kappa_{\infty}}$, which leads to
	\begin{equation}\label{Lemma_B_3_tii}
	\begin{split}
	&\|(\nabla^2_{SS}L_n(\by,\overline{\bW}\bar{\btheta}^*))^{-1}\|_{\infty}\leq\|(\nabla^2_{SS}L_n(\by,\bW\btheta^*))^{-1}\|_{\infty}+\frac{1}{4\kappa_{\infty}}\leq \frac{1}{2\kappa_{\infty}}.\\
	\end{split}
	\end{equation}
	
	\medskip	
	
	(iii)	The third argument follows (\ref{Lemma_B_3_tii}) easily. Since $\|\bA\|_2\leq\|\bA\|_{\infty}$ holds for any symmetric matrix $\bA$, we have $\|(\nabla^2_{SS}L_n(\by,\overline{\bW}\bar{\btheta}^*))^{-1}-(\nabla^2_{SS}L_n(\by,\bW\btheta^*))^{-1}\|_2
	\leq \frac{1}{4\kappa_{\infty}}
	\leq \frac{1}{4\kappa_2}$
	and thus $\|(\nabla^2_{SS}L_n(\by,\overline{\bW}\bar{\btheta}^*))^{-1}\|_2\leq\frac{1}{2\kappa_2}$.
	
	\medskip
	
	(iv) Finally we prove the last inequality. On the one hand,
	\begin{equation*}
	\begin{split}
	&\|\nabla^2_{S_2 S} L_n(\by,\overline{\bW}\bar{\btheta}^{\ast})(\nabla_{S S}^2 L_n(\by,\overline{\bW}\bar{\btheta}^{\ast}))^{-1}
	-\nabla^2_{S_2 S} L_n(\by,\bW\btheta^{\ast})(\nabla_{S S}^2 L_n(\by,\bW\btheta^{\ast}))^{-1}\|_{\infty}\\
	&\leq \|\nabla^2_{S_2 S} L_n(\by,\overline{\bW}\bar{\btheta}^*)-\nabla^2_{S_2 S} L_n(\by,\bW\btheta^{\ast})\|_{\infty}
	\|(\nabla_{S S}^2 L_n(\by,\overline{\bW}\bar{\btheta}^*))^{-1}\|_{\infty}\\
	&+\|\nabla^2_{S_2 S} L_n(\by,\bW\btheta^{\ast})[(\nabla_{S S}^2 L_n(\by,\overline{\bW}\bar{\btheta}^*))^{-1}-(\nabla_{S S}^2 L_n(\by,\bW\btheta^{\ast}))^{-1}]\|_{\infty}.
	\end{split}
	\end{equation*}
From claim (ii) and (\ref{lem-B-inf}) it is easy to see that
	\begin{equation*}
	\begin{split}
	&\|\nabla^2_{S_2 S} L_n(\by,\overline{\bW}\bar{\btheta}^*)-\nabla^2_{S_2 S} L_n(\by,\bW\btheta^{\ast})\|_{\infty}
	\|(\nabla_{S S}^2 L_n(\by,\overline{\bW}\bar{\btheta}^*))^{-1}\|_{\infty}\leq\frac{1}{4\kappa_{\infty}}3M_0M_2|S|\delta.
	\end{split}
	\end{equation*}
	On the other hand, we can take $\bA=\nabla^2_{S_2 S} L_n(\by,\bW\btheta^{\ast})$, $\bB=\nabla_{S S}^2 L_n(\by,\bW\btheta^{\ast})$ and $\bC=\nabla_{S S}^2 L_n(\by,\overline{\bW}\bar{\btheta}^*)-\nabla_{S S}^2 L_n(\by,\bW\btheta^{\ast})$. By Assumption \ref{assump-RSC-1}, $\|\bA\bB^{-1}\|_{\infty}\leq 1-2\tau\leq 1$. Lemma \ref{inverse-perturbation} forces that	
	\begin{equation*}
	\begin{split}
	&\|\nabla^2_{S_2 S} L_n(\by,\bW\btheta^{\ast})[(\nabla_{S S}^2 L_n(\by,\overline{\bW}\bar{\btheta}^*))^{-1}-(\nabla_{S S}^2 L_n(\by,\bW\btheta^{\ast}))^{-1}]\|_{\infty}\\
	&=\|\bA[(\bB+\bC)^{-1}-\bB^{-1}]\|_{\infty}\leq \|\bA\bB^{-1}\|_{\infty} \frac{\|\bC\bB^{-1}\|_{\infty}}{1-\|\bC\bB^{-1}\|_{\infty}}\leq \frac{\|\bC\|_{\infty}\|\bB^{-1}\|_{\infty}}{1-\|\bC\|_{\infty}\|\bB^{-1}\|_{\infty}}.
	\end{split}
	\end{equation*}
	We have shown above in (\ref{lemma_B_3_alpha}) that $\|\bC\|_{\infty}\|\bB^{-1}\|_{\infty}\leq\frac{3}{8\kappa_{\infty}}M_0M_2|S|\delta\leq 1/2$.
	As a result,
	\begin{equation*}
	\begin{split}
	&\|\nabla^2_{S_2 S} L_n(\by,\bW\btheta^{\ast})[(\nabla_{S S}^2 L_n(\by,\overline{\bW}\bar{\btheta}^*))^{-1}-(\nabla_{S S}^2 L_n(\by,\bW\btheta^{\ast}))^{-1}]\|_{\infty}
	\leq \frac{3}{4\kappa_{\infty}}M_0M_2|S|\delta.
	\end{split}
	\end{equation*}
	
	By combining these estimates, we have
	\begin{equation*}
	\begin{split}
	&\|\nabla^2_{S_2 S} L_n(\by,\overline{\bW}\bar{\btheta}^*)(\nabla_{S S}^2 L_n(\by,\overline{\bW}\bar{\btheta}^*))^{-1}-\nabla^2_{S_2 S} L_n(\by,\bW\btheta^{\ast})(\nabla_{S S}^2 L_n(\by,\bW\btheta^{\ast}))^{-1}\|_{\infty}\\
	&\leq  \frac{3}{2\kappa_{\infty}}M_0M_2|S|\delta\leq \tau.
	\end{split}
	\end{equation*}
	Therefore $\|\nabla^2_{S_2 S} L_n(\by,\overline{\bW}\bar{\btheta}^*)(\nabla_{S S}^2 L_n(\by,\overline{\bW}\bar{\btheta}^*))^{-1}\|_{\infty}\leq (1-2\tau)+\tau=1-\tau$.
\end{proof}

\subsection{Proof of Lemma \ref{lem-epsilon}}
	
	Let $ \varepsilon_j=\frac{1}{n}\sum_{t=1}^{n}\overline{w}_{tj} \eta_t$ for $j\in[p+K]$.
	Observe that $\varepsilon=\max_{j\in[p+K]} \varepsilon_j$ and
	$\varepsilon_j=
	| \frac{1}{n}\sum_{t=1}^{n} w_{tj}\eta_t
	+ \frac{1}{n}\sum_{t=1}^{n} (\overline{w}_{tj} - w_{tj} ) \eta_t |$. By Cauchy-Schwarz inequality,
	\begin{equation*}
	\left| \frac{1}{n}\sum_{t=1}^{n} (\overline{w}_{tj} - w_{tj} ) \eta_t \right|
	\leq \left(\frac{1}{n}\sum_{t=1}^{n} (\overline{w}_{tj} - w_{tj} )^2 \right)^{1/2}
	\left(\frac{1}{n}\sum_{t=1}^{n} \eta_t^2 \right)^{1/2}.
	\end{equation*}
	As a result,
	\begin{equation}
	\varepsilon \leq \max_{j\in[p+K]}\left|
	\frac{1}{n}\sum_{t=1}^{n} w_{tj}\eta_t
	\right|+
	\left(\frac{1}{n}\sum_{t=1}^{n} \eta_t^2 \right)^{1/2}
	\max_{j\in[p+K]} \left(\frac{1}{n}\sum_{t=1}^{n} (\overline{w}_{tj} - w_{tj} )^2 \right)^{1/2}.
	\label{ineq-epsilon}
	\end{equation}
	
	By Theorem 1 and Remark 1 in \cite{MPR11}, there exist constants $C_1$, $C_2$, $C_3$ and $C_4$ such that for any $s\geq 0$,
	\begin{align*}
	\mathrm{P}\left(
	\left| \frac{1}{n}\sum_{t=1}^{n} w_{tj}\eta_t \right| >s
	\right) \leq n\exp\left(
	-\frac{(ns)^{\gamma}}{C_1}
	\right) + \exp\left(
	-\frac{(ns)^2}{nC_2}
	\right) + \exp\left(
	-\frac{(ns)^2}{C_3n} \exp\left(
	\frac{(ns)^{\gamma(1-\gamma)}}{ C_4 [\log (ns)]^{\gamma} }
	\right)\right).
	\end{align*}
	From this it is easily seen that for large enough constant $C>0$, we have
	$\mathrm{P}\left(
	\left| \frac{1}{n}\sum_{t=1}^{n} w_{tj}\eta_t \right| >C\sqrt{\frac{\log p}{n}}
	\right) \leq p^{-2}$ for all $j\in[p+K]$. Union bounds then force the first in (\ref{ineq-epsilon}) to be of order $O_{\mathrm{P}} (\sqrt{\frac{\log p}{n}})$.
	
	Similarly, we can apply the concentration inequality in \cite{MPR11} to get $\frac{1}{n}\sum_{t=1}^{n} \eta_t^2 = O_{\mathrm{P}} (1)$. It follows from Lemma \ref{lem-factor} that $\max_{j\in[p+K]} \left[\frac{1}{n}\sum_{t=1}^{n} (\overline{w}_{tj} - w_{tj} )^2 \right]^{1/2} = O_{\mathrm{P}} ( \sqrt{\frac{\log p}{n}} + \frac{1}{\sqrt{p}} )$. Hence the second term in (\ref{ineq-epsilon}) is of order $O_{\mathrm{P}} ( \sqrt{\frac{\log p}{n}} + \frac{1}{\sqrt{p}} )$.

\end{document}